\documentclass{article}
\usepackage{ijcai25}

\usepackage{times}
\usepackage{soul}
\usepackage{url}
\usepackage[hidelinks]{hyperref}
\usepackage[utf8]{inputenc}
\usepackage[small]{caption}
\usepackage{graphicx}
\usepackage{amsmath,amsthm}
\usepackage{booktabs}
\usepackage{algorithm}
\usepackage{algorithmic}
\usepackage[switch]{lineno}

\newcommand\blfootnote[1]{%
  \begingroup
  \renewcommand\thefootnote{}\footnote{#1}%
  \addtocounter{footnote}{-1}%
  \endgroup
}

\newtheorem{theorem}{Theorem}

\newif\iflong
\newif\ifshort

 \longtrue

\iflong
\else
\shorttrue
\fi

\title{A Structural Complexity Analysis of Hierarchical Task Network Planning}

\author{
Cornelius Brand$^1$
\and
Robert Ganian$^2$
\and
Fion Mc Inerney$^3$
\and
Simon Wietheger$^2$\\
\affiliations
$^1$Algorithms \& Complexity Theory Group, Regensburg University, Germany\\
$^2$Algorithms and Complexity Group, TU Wien, Austria\\
$^3$Telefónica Scientific Research, Barcelona, Spain\\
\emails
cornelius.brand@ur.de, \{rganian, fmcinern\}@gmail.com, swietheger@ac.tuwien.ac.at
}

\usepackage{diagbox}
\usepackage[dvipsnames]{xcolor}
\usepackage{amsmath,amsfonts,amsthm,amssymb}
\usepackage{todonotes}

\usepackage{tabularx, environ}
\usepackage{tcolorbox}
\usepackage{makecell}
\usepackage{tablefootnote}
\usepackage[capitalise]{cleveref}
\usepackage{xspace}
\usepackage{enumitem}
\usepackage{verbatim}

\usepackage{tikz}
\usetikzlibrary{shapes}
\usetikzlibrary{arrows,automata, positioning,calc, patterns,decorations.markings}

\newtheorem{lemma}[theorem]{Lemma}

\theoremstyle{definition}
\newtheorem{definition}{Definition}
\newtheorem{observation}{Observation}

\newtheorem*{claim*}{Claim}
\newtheorem{claim}{Claim}

\newcommand{\bigO}[1]{{\ensuremath{\mathcal{O}\!\left(#1\right)}}\xspace}

\newcommand{\N}{\mathbb{N}}

\newcommand{\C}{{\cal C}}

\newcommand{\G}{{\cal G}}
\newcommand{\I}{{\cal I}}

\newcommand{\NP}{{\normalfont\textsf{NP}}\xspace}
\newcommand{\W}{{\normalfont\textsf{W}}}

\newcommand{\FPT}{{\normalfont\textsf{FPT}}\xspace}

\newcommand{\calS}{{\cal S}}

\newcommand{\set}[1]{{\{#1\}}}

\newcommand{\bigoh}{\mathcal{O}}

\newcommand{\poly}{\text{P}\xspace}

\newcommand{\np}{\text{NP}\xspace}

\newcommand{\nphard}{\NP-hard\xspace}
\newcommand{\npcomplete}{\NP-complete\xspace}

\newcommand{\PisNP}{\(\poly =\np\)\xspace}

\newcommand{\pverify}{{\sc Plan Verification}\xspace}
\newcommand{\ptarget}{{\sc State Reachability}\xspace}
\newcommand{\pexist}{{\sc Plan Existence}\xspace}
\newcommand{\pcover}{{\sc Action Executability}\xspace}

\newcommand{\compn}{{\ensuremath{C_\#}}\xspace}
\newcommand{\comps}{{\ensuremath{C_s}}\xspace}
\newcommand{\compd}{{\ensuremath{C_d}}\xspace}
\newcommand{\compc}{{\ensuremath{C_c}}\xspace}

\newcommand{\compcdns}{\compc,\compd,\compn,\comps}
\newcommand{\compdns}{\compd,\compn,\comps}

\DeclareMathOperator{\pre}{prec}
\DeclareMathOperator{\del}{del}
\DeclareMathOperator{\add}{add}

\newcommand{\True}{TRUE\xspace}
\newcommand{\False}{FALSE\xspace}

\newcommand{\problem}[3]{%
      \noindent
      \begin{center}
      \begin{tcolorbox}[title=~#1,left=-1.5mm,top=0mm,bottom=0mm,right=0mm,boxsep=1mm]
        \begin{tabular}{p{.15\textwidth}p{.8\textwidth}}
            \textbf{Input:} & \parbox[t]{.76\textwidth}{#2}\\[15pt]
            \textbf{Question:} & \parbox[t]{.76\textwidth}{#3}\\
        \end{tabular}
      \end{tcolorbox}
      \end{center}     
      }

\definecolor{cb_orange}{rgb}{0.961, 0.475, 0.227}
\definecolor{cb_purple}{rgb}{0.663, 0.353, 0.631}
\definecolor{cb_cyan}{rgb}{0.522, 0.753, 0.976}
\definecolor{cb_blue}{rgb}{0.059, 0.125, 0.502}

\newcommand{\citey}[1]{\shortcite{#1}}

\begin{document}

\maketitle

\begin{abstract}
    We perform a refined complexity-theoretic analysis of three classical problems in the context of Hierarchical Task Network Planning: the verification of a provided plan, whether an executable plan exists, and whether a given state can be reached. Our focus lies on identifying structural properties which yield tractability. We obtain new polynomial algorithms for all three problems on a natural class of primitive networks, along with corresponding lower bounds. We also obtain an algorithmic meta-theorem for lifting polynomial-time solvability from primitive to general task networks, and prove that its preconditions are tight.
    Finally, we analyze the parameterized complexity of the three problems. 
\blfootnote{In line with IJCAI standards and guidelines, full details for some proofs are deferred to the technical appendix.}
\end{abstract}

\section{Introduction}
Automated Planning is a core research topic whose goal is to devise computational methods capable of finding solutions to prominent planning tasks. In this work, we consider automated planning in the context of {\it Hierarchical Task Networks (HTNs)}, which have received significant attention in the artificial intelligence research community.
HTNs are capable of incorporating the compound, hierarchical structure of real-world tasks, as opposed to the fine-grained view of classical planning that operates directly on the state space of the system under consideration, and
are one of the best-established planning formalisms in the  literature~\cite{TsunetoEHN96,LiKY09,SohrabiBM09,Geier_Bercher_2011,Alford2015TightBF,XiaoHPWS17,BehnkeHB19,HollerBBB19,Lin_Bercher_2023}.

On a high level,
HTNs can be informally described and illustrated as follows:
at the top level of the task hierarchy is a set of tasks that is to be executed in a (partially specified) order. For instance, the task of going to work may be decomposed into three subtasks: leaving the house, going the actual distance, and getting set up at the office. Now, each of these tasks can in turn be decomposed into a set of smaller tasks,
like climbing the stairs or taking a bus, which again can be decomposed, and so forth. 
This continues until we reach an elementary level of tasks that are no longer decomposable,
which correspond to single executable \emph{actions}.

The crucial feature of the model is that it separates the user-facing description of a task (what we want to do) from its internal implementation (how to do it). For instance, leaving the house needs to be done \emph{before} going the actual distance, regardless of how it is done.
However, in an HTN, we may have {\it several} options of how to decompose this task: the agent can, e.g., either climb the stairs or take the elevator. 

Perhaps the most natural problem associated with HTNs is \pverify: can the tasks in a given network be decomposed and then ordered to match a given, executable sequence of actions?
A second fundamental problem is \pexist: can the tasks in a given network be decomposed and then ordered in any way that yields an executable sequence of actions?
The literature has investigated the complexity of several variants of the aforementioned problems~\cite{Erol_Hendler_Nau_1996,Geier_Bercher_2011,Alford2015TightBF,Behnke_Hoeller_Biundo_2015,Lin_Bercher_2023}.
Moreover, we also consider the question of whether a specified target state can be reached, formalized as the \ptarget\ problem.
This is the typical goal in classical planning \cite{Fikes_Nilsson_1971,Chapman_1987,Backstrom_Nebel_1995}, and has also been considered in the HTN setting~\cite{Hoeller_Bercher_2021,Olz_Biundo_Bercher_2021}.

It is well-established that problems on HTNs such as those listed above
remain computationally intractable even in severely restricted settings.
These results usually depend on restricting the decompositions, e.g., requiring acyclicity \cite{Alford_Bercher_Aha_2015,Chen_Bercher_2021}, restricting the impact of actions on the state space (e.g., delete-relaxed HTN planning)~\cite{Alford_Shivashankar_Kuter_Nau_2014,Hoeller_Bercher_Behnke_2020}, or by allowing the insertion of tasks at arbitrary positions \cite{Geier_Bercher_2011,Alford2015TightBF}.

In this article, we take a different approach and instead focus on structural restrictions of the task network.
While there are some complexity results known for the extreme cases of a total ordering or no ordering of the tasks \cite{Geier_Bercher_2011,Alford_Shivashankar_Kuter_Nau_2014,Behnke_Speck_2021}, other restrictions received little to no attention so far.
Building on these first results, we identify interesting and natural special cases where pockets of polynomial-time solvability show themselves. Moreover, we prove that these pockets give rise to a surprisingly rich complexity-theoretic landscape. 

\paragraph*{Our Contributions.}
We first consider \emph{primitive} task networks, i.e., HTNs in which there are no compound tasks. While this may seem restrictive, it is a natural place to start from a complexity-theoretic perspective: all three of the considered problems remain \NP-hard on primitive task networks. 
In fact, \pverify is known to be \nphard even if the network forms a collection of disjoint stars \cite{Lin_Bercher_2023}.
Moreover, unlike \pverify, the latter two problems remain \NP-hard even on trivial (in particular, edgeless) networks when the state space is not of fixed size \cite{Bylander94,Erol_Hendler_Nau_1996}.
Thus, in line with previous related works, we henceforth consider \pexist\ and \ptarget\ exclusively in the setting of constant-sized state spaces~\cite{Kronegger_Pfandler_Pichler_2013,Backstrom_Jonsson_Ordyniak_Szeider_2015}.
Even in this setting, we establish the \NP-hardness of the three considered problems on networks consisting of a collection of independent chains.

While these initial results may seem discouraging, it turns out that all three problems are polynomial-time solvable on networks which are chains (i.e., totally-ordered networks)~\cite{Erol_Hendler_Nau_1996,Behnke_Hoeller_Biundo_2015} or antichains (\cref{thm:reachtarget_fpt_empty,thm:includeaction_fpt_empty}).
We generalize and unify these tractability results by showing that all three considered problems are polynomial-time solvable on networks of bounded ``generalized partial order width'', that is, the partial order width of the network when ignoring isolated elements (\Cref{thm:planverif_powidth_xp,thm:target_cover_powidth_xp}). 

Building on the above results for the primitive case, in the second part of the paper, we establish an \emph{algorithmic meta-theorem} that allows us to precisely describe the conditions under which polynomial-time tractability
for primitive task networks can be lifted to the compound case (\Cref{thm:meta_for_compounds}).  
Crucially, we prove that this result is tight in the sense that all of the required conditions are \emph{necessary} (\Cref{thm:meta_theorem_lowerbound}).
In the final part of the paper, we push beyond classical complexity and analyze the considered problems from the perspective of the more fine-grained \emph{parameterized complexity} paradigm~\cite{CyganFKLMPPS15}. There, one analyzes problems not only w.r.t.\ the input size $n$, but also w.r.t.\ a specified parameter $k\in \N$, where the aim is to obtain algorithms which run in time at most $f(k)\cdot n^{\bigoh(1)}$ for some computable function $f$ (so-called \emph{fixed-parameter algorithms}); problems admitting such algorithms are called \emph{fixed-parameter tractable} (\FPT). 

The natural question from the parameterized perspective is whether the aforementioned tractability results for HTNs with constant-sized structural measures (\Cref{thm:planverif_powidth_xp,thm:target_cover_powidth_xp}) can be lifted to fixed-parameter tractability when parameterized by the same measures. While analogous questions have been thoroughly explored---and often answered in the affirmative---in many other areas of artificial intelligence~\cite{EibenGKOS21,GanianK21,FroeseKN22,HeegerHMMNS21},
here we prove that the generalized partial order width does not yield fixed-parameter tractability for any of our problems on primitive task networks, under well-established complexity assumptions. Moreover, 
the aforementioned hardness of the problems on forests of stars already rules out fixed-parameter tractability for essentially all standard structural graph parameters, except for the \emph{vertex cover number} (vcn)~\cite{BlazejGKPSS23,Icalpvc,FionnColt}. We complement these lower bounds by designing highly non-trivial fixed-parameter algorithms for all three considered problems on primitive task networks w.r.t.\ the vertex cover number of the network. For compound HTNs, we then adapt our general algorithmic meta-theorem for polynomial-time solvability to the fixed-parameter paradigm (\Cref{thm:meta_for_compounds_fpt}), allowing us to also extend these results to the non-primitive setting.
See \cref{fig:mindmap} for a schematic overview of our results.
\begin{figure}
    \tikzstyle{block} = [rectangle, draw, fill=white!20, 
        text width=8.1em, text centered, rounded corners, minimum height=1em]
    \tikzstyle{line} = [draw, -latex']
    \tikzstyle{cloud} = [draw, text centered, rectangle,fill=gray!10, node distance=0cm, text width=6.5em,
        minimum height=0.0em]
    \tikzstyle{cloudhard} = [draw, text centered, rectangle,fill=gray!45, node distance=0cm, text width=6.5em,
        minimum height=0.0em]
        \centering
        \begin{tikzpicture}[node distance =0cm, auto]
        \small
        \centering
        %
        %
        \node [block, fill=black!20] (hardcases) {\scriptsize Hardness on Tree-like Graphs (\ref{thm:planverif_chains_hard}, \ref{thm:target_cover_chains_hard})};
        %
        %
        \node [block, right of=hardcases, node distance=2.85cm, fill=black!20] (chainAntichain) {\scriptsize Tractability on (Anti)-Chains (\ref{thm:reachtarget_fpt_empty}, \ref{thm:includeaction_fpt_empty})};
        %
        %
         \node [block, below of=chainAntichain, xshift=-2.3cm, node distance=1.2cm,, fill=black!20] (algoPOW) {\scriptsize P for fixed gpow (\ref{thm:planverif_powidth_xp},\ref{thm:target_cover_powidth_xp})};
         %
         %
         \node [block, right of=algoPOW, node distance=5.25cm] (meta) {\scriptsize Meta-Theorem (\ref{thm:meta_for_compounds})};
         %
         %
         \node [block, right of=hardcases, node distance=5.8cm] (lowerboundmeta) {\scriptsize Lower Bound for Meta-Theorem (\ref{thm:meta_theorem_lowerbound})};
         %
         %
         \node [block, below of=algoPOW, node distance=1.4cm, fill=black!20] (primVCN) {\scriptsize \FPT w.r.t.\ vcn (\ref{thm:planverif_vertexcover}, \ref{thm:reachtarget_includeaction_fpt})};
         %
         %
         \node [cloud, right of=algoPOW, node distance=2.625cm, yshift=-0.7cm,fill=black!0,text width=6.9em] (compoundAlgos) {\scriptsize Algorithms for Compound Networks};
        %
        %
         \node [block, right of=primVCN, node distance=5.25cm] (meta_fpt) {\scriptsize \FPT\ Meta-Theorem (\ref{thm:meta_for_compounds_fpt})};
            
            \draw[<->] (meta) -- (lowerboundmeta);

            \draw[->] (algoPOW) -| (compoundAlgos);
            \draw[->] (meta) -| (compoundAlgos);
           
            \draw[->] (primVCN) -| (compoundAlgos);
            \draw[->] (meta_fpt) -| (compoundAlgos);

            \draw[->] (chainAntichain.south) -- +(0cm,-0.205cm) -| node [below, xshift=.5cm] {\tiny motivate}(algoPOW.north);
            \draw[->] (hardcases.south) -- +(0cm,-0.2cm) -| (algoPOW.north);
            \draw[->] ([xshift=-1cm]hardcases.south) |-  node [rotate=-90, xshift=-1.2cm, yshift=-0.4cm] {\tiny motivates} (primVCN.west);
       
        \end{tikzpicture}
    \caption{Overview of our results and respective theorems. Generalized Partial Order Width is abbreviated to \emph{gpow}, and Vertex Cover Number to \emph{vcn}. Results in gray boxes refer to primitive instances.} 
    \label{fig:mindmap}
    \end{figure}

\paragraph*{Related Work.}
HTN planning has been extensively studied in artificial intelligence research and has applications in areas such as 
healthcare~\cite{Gonzalez-FerrerTFM13}, business processes~\cite{Gonzalez-FerrerFC13}, human-robot interaction~\cite{HayesS16}, 
emergency response~\cite{LiuWQZW16,ZhaoQL17}, and visualization~\cite{PadiaBH19}. 
For an overview of the current state of HTN planning, see the works by Bercher, Alford, and Hoeller~\citey{Bercher_Alford_Hoeller_2019} and  Georgievski and Aiello~\citey{Georgievski_Aiello_2015}.

Questions on the complexity and even computability of many tasks associated with HTNs have been treated in the literature~\cite{ErolHN94,Erol_Hendler_Nau_1996,Geier_Bercher_2011,Alford_Bercher_Aha_2015,Olz_Biundo_Bercher_2021,Lin_Bercher_2023,LinHB24,LinOHB24}.
Albeit carried out in a different context, the work of Kronegger, Pfandler, and Pichler~\citey{Kronegger_Pfandler_Pichler_2013} on propositional STRIPS with negation (PSN) studies a problem ($k$-PSN) which is related to \ptarget. The main difference is that while PSN includes a broader definition of actions, it does not consider any primitive or hierarchical network structure.
Finally, the computational complexity of planning has also been studied in settings beyond HTNs~\cite{Bylander94,BackstromCJOS12,BackstromJOS13,Backstrom_Jonsson_Ordyniak_Szeider_2015}.

\section{Preliminaries}
\label{sec:prelims}
For $n\in \mathbb{N}$, 
let \([n]=\set{1,\ldots,n}\). We use standard graph terminology~\cite{Diestel}. 
The \emph{partial order width} $w$ of a partially ordered set $\lesssim$ is  
the maximum number of pairwise-disjoint maximal chains contained in $\lesssim$ or, equivalently, the size of its largest antichain. An element in $\lesssim$ is \emph{isolated} if it is incomparable to each other element in the set. We define the \emph{generalized partial order width} of a partially ordered set $\lesssim$ as the partial order width of $\lesssim$ after removing isolated elements.

\paragraph*{Hierarchical Task Networks.}
Our formalization of HTNs follows the terminology employed in recent works~\cite{Olz_Biundo_Bercher_2021,Lin_Bercher_2023}.
The {\it domain} of a planning problem is a tuple \((F, A, \C, \delta, M)\),
where \(F\) is a finite set of propositions, \(A\) is a set of {\it actions} (or {\it primitive task names}), and \(\C\) is a set of {\it compound task names} such that \(A\cap \C = \emptyset\).
The function \(\delta \colon A \rightarrow 2^F\times 2^F\times 2^F\) maps each action to a set \(\pre(a)\) of preconditions, a set \(\del(a)\) of propositions to be removed, and a set \(\add(a)\) of propositions to be added to the current state.
An action \(a\) is {\it executable} in a state \(s\subseteq F\) if \(\pre(a)\subseteq s\).
If executed, it changes the state \(s\) to \(s' = (s\setminus \del(a)) \cup \add(a)\).
A \emph{plan} is a sequence of actions and is executable in a state $s$ if its actions can be executed consecutively starting from \(s\).
Lastly, $M$ represents the \emph{decomposition methods} and maps each $c\in \C$ to a set of task networks (defined below) which $c$ can be decomposed into.

A {\it task network} in such a domain is a tuple \((T, \prec^+, \alpha)\) consisting of a set \(T\) of task identifiers (or simply \emph{tasks}) with a partial order \(\prec^+\ \subseteq T\times T\) and a function \(\alpha\colon T \rightarrow A\cup \C\).

The usual representation of a task network is its cover graph  \(D_\prec=(T,\prec)\), where $\prec$ is the cover of \(\prec^+\) (in particular, $(a,b)\in \prec$ if and only if $a \prec^+ b$ and there is no $c$ such that $a \prec^+ c \prec^+ b$). We denote the underlying undirected graph by $G_\prec$.
\cref{fig:example_network} provides an example of a task network.

\begin{figure}
    \centering
    \begin{minipage}{.18\textwidth}
        \includegraphics[width=.8\textwidth]{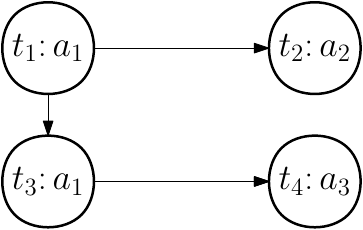}    
    \end{minipage}
    \hfill
    \begin{minipage}{.27\textwidth}
        \begin{tabular}{|c|c|c|c|} 
            \hline
            Action & $\pre$ & $\del$ & $\add$ \\ 
            \hline
            $a_1$ & $\emptyset$ & $\emptyset$ & $\set{1}$ \\ 
            \hline
            $a_2$ & $\emptyset$ & $\emptyset$ & $\set{2}$ \\ 
            \hline
            $a_3$ & $\set{2} $ & $\set{1}$ & $\emptyset$ \\ 
            \hline
            \end{tabular}
    \end{minipage}
  \caption{The digraph $D_\prec$ for a primitive HTN with tasks $T=\set{t_1,t_2,t_3,t_4}$, actions $A=\set{a_1,a_2,a_3}$, propositions $F=\{1,2\}$, and initial state \(s_0 = \emptyset\). 
  A task's action is in its node.
  There is an arc from a task $t$ to a task $t'$ if $(t,t')\in\prec$.
  Here, the only sequences that execute all tasks are $t_1,t_2,t_3,t_4$ and $t_1,t_3,t_2,t_4$.
  }
  \label{fig:example_network}     
\end{figure}
A subset \(T'\) of tasks is a {\it chain} if \(\prec^+\) defines a total order on \(T'\), and an {\it antichain} if all of its tasks are pairwise incomparable.
A task \(t\) is {\it primitive} if \(\alpha(t)\in A\), and {\it compound} otherwise.
A task network is {\it primitive} if it contains no compound tasks.
For a primitive task network, a total order $\bar{t}$ on \(T\) respecting \(\prec\) is called a {\it linearization} and we denote the corresponding action sequence by $\alpha(\bar{t})$.
Non-primitive task networks can be decomposed into primitive ones via the methods in \(M\); decomposing a task \(t\) into a subnetwork \(tn_m\in M(t)\) means replacing \(t\) by \(tn_m\) such that each task in \(tn_m\) is connected to tasks outside \(tn_m\) in the same way that \(t\) was.

The input of an {\it HTN problem} contains a domain \(D=(F, A, \C, \delta, M)\), an initial task network \(tn\) in that domain, and an initial state \(s_0 \subseteq F\).\footnote{
    Some works refer to the triple $(D, tn, s_0)$ as an \emph{HTN problem}, while we use that term for the considered algorithmic problems.}
In a generalization of the literature, we consider a primitive task network $tn' = (T', {\prec^+}', \alpha')$ to be a \emph{solution} to a triple $(D,tn,s_0)$ if $tn'$ has a linearization~$\bar{t}$ such that $\alpha'(\bar{t})$ is executable in $s_0$ and $tn$ decomposes into a primitive task network $(T^*, {\prec^+}^*, \alpha^*)$ with
 $T'\subseteq T^*$,
  for each $(t,t')\in {\prec^+}^*$ with $t'\in T'$ we have that $t\in T'$ and $(t,t')\in {\prec^+}'$,
   and $\alpha'(t)=\alpha^*(t)$ for all $t\in T'$.
The literature usually considers the case where $tn$ can be decomposed into $tn'$, that is, $T^*=T'$, in which we call $tn'$ a \emph{full solution}.
 Intuitively, non-full solutions respect $\delta$ and $\prec^+$, but do not need to execute all tasks.
We now define our problems of interest.
%
      \begin{center}
      \begin{tcolorbox}[title=~\pverify,left=-1.5mm,top=0mm,bottom=0mm,right=0mm,boxsep=1mm]
        \begin{tabular}{p{.15\textwidth}p{.8\textwidth}}
            \textbf{Input:} & \parbox[t]{.76\textwidth}{A task network \(tn\) in a domain $D$,
            an initial state \(s_0\subseteq F\), and a plan \(p\) over \(A\).}\\[15pt]
            \textbf{Question:} & \parbox[t]{.76\textwidth}{Is there a full solution to $(D, tn, s_0)$ with a linearization $\bar{t}$ such that $\alpha(\bar{t})=p$?}\\
        \end{tabular}
      \end{tcolorbox}
      \end{center}     

\problem{\pexist}%
{A task network \(tn\) in a domain $D$ and an initial state \(s_0\subseteq F\).}%
{Is there a full solution to $(D, tn, s_0)$?}

In addition to the problems mentioned in the introduction, we also define the corresponding problem of constructing a plan executing (at least) a specified set of actions. This \pcover\ problem is closely related to \pexist\ and, in fact, all results that we prove for one of the problems will also hold for the other.

      \noindent
      \begin{center}
      \begin{tcolorbox}[title=~\pcover,left=-1.5mm,top=0mm,bottom=0mm,right=0mm,boxsep=1mm]
        \begin{tabular}{p{.15\textwidth}p{.8\textwidth}}
            \textbf{Input:} & \parbox[t]{.76\textwidth}{A task network \(tn\) in a domain $D$, an initial state \(s_0\subseteq F\), and a multiset \(S\) of actions in~\(A\).}\\[15pt]
            \textbf{Question:} & \parbox[t]{.76\textwidth}{Is there a solution to $(D, tn, s_0)$ with a linearization $\bar{t}$ such that $\alpha(\bar{t})$ covers $S$?}\\
        \end{tabular}
      \end{tcolorbox}
      \end{center}     
\pcover can be seen as a generalization of \pexist: any \pexist instance can be easily reduced to \pcover, e.g., by adding a new task \(t\) with new action \(a\), letting \(S=\set{a}\), and adapting $\prec^+$ or $\delta$ to ensure that all other tasks are executed before $t$.
Such direct reductions can be used to transfer algorithmic and hardness results in the respective directions; in our setting, we establish all of our hardness results for the ``simpler'' \pexist, while all of our algorithms can deal with the ``harder'' \pcover.

      \noindent
      \begin{center}
      \begin{tcolorbox}[title=~\ptarget,left=-1.5mm,top=0mm,bottom=0mm,right=0mm,boxsep=1mm]
        \begin{tabular}{p{.15\textwidth}p{.8\textwidth}}
            \textbf{Input:} & \parbox[t]{.76\textwidth}{A task network \(tn\) in a domain $D$, an initial state \(s_0\subseteq F\), and a target state \(s_g\subseteq F\).}\\[15pt]
            \textbf{Question:} & \parbox[t]{.76\textwidth}{Is there a solution to $(D, tn, s_0)$ with a linearization yielding a state \(s \supseteq s_g\)?}\\
        \end{tabular}
      \end{tcolorbox}
      \end{center}     

While the considered problems are stated in their decision variants
for complexity-theoretic purposes, all of our obtained algorithms are constructive and can output a solution and a corresponding linearization as a witness. 
Further, for primitive networks, all the problems are clearly in \NP. 
In some results, we use the {\it state transition graph}, i.e., the multigraph whose vertices are the states reachable from the initial state~$s_0$ by any sequence of actions, and a transition from one state to another via an action $a$ is represented by an arc labeled $a$.

\section{Primitive Task Networks}\label{sec:primitiveNetworks}
Our first set of results focuses on the classical complexity of the targeted problems on primitive task networks. Before proceeding with our analysis, we make a general observation about the state space in our instances.

\paragraph*{The State Space.}
\label{sec:states}
As a basic precondition for solving our problems on more sophisticated task networks, we need to ensure tractability on trivial primitive task networks where $\prec^+ =\emptyset$ (antichains). 
Unfortunately, all but the first of our problems of interest are intractable even on these degenerate networks, as the hardness of \pexist in STRIPS planning \cite{Bylander94} transfers to HTN planning by an HTN-to-STRIPS translation \cite{AlfordBHBBA16}; see also the related work of Nebel and Bäckström~\citey{Nebel_Backstrom_1994} and Erol et al.~\citey{Erol_Hendler_Nau_1996}. 
As discussed above, the hardness of \pcover follows from the hardness of \pexist.
The hardness of \ptarget follows as well: there is a trivial reduction from \pexist that adds a special task $t_g$ and adapts all actions and $s_g$ such that $s_g$ is only reachable if $t_g$ is executed, and $t_g$ is only executable after all other tasks have been executed.

A requirement for the above reductions
is that the state space (i.e., the number of states in the state transition graph)
can be large, and in particular not bounded by any constant. 
Hence, we hereinafter analyze the affected problems (\ptarget, \pexist, and \pcover) in the bounded state-space setting.

\paragraph*{The Intractability of Tree-like HTNs.}
\label{sub:trees}
We begin with a series of lower bounds which rule out polynomial algorithms even for simple ``tree-like'' networks (unless \PisNP).
\pverify\ is known to be 
\npcomplete\ on primitive networks such that $D_\prec$ consists of disjoint stars, even when \(\pre(a)=\del(a)=\add(a)=\emptyset\) for all actions \(a\in A\)~\cite[Theorem~3]{Lin_Bercher_2023}.
As our first result of this section, we show that our problems remain intractable even if \(\prec^+\) describes a set of total orders that are independent from one another (i.e., $D_\prec$ is a disjoint union of paths).
We reduce from \textsc{Shuffle Product}, which asks, given words \(u,c_1,\ldots,c_w\) over an alphabet \(\Sigma\), to decide if $u$ can be generated by interlacing \(c_1,\ldots,c_w\) while preserving the order of the letters as they appear in each $c_i$, $i\in [w]$. This problem is \nphard even if \(|\Sigma|=2\)~\cite{Warmuth_Haussler_1984}.
\begin{theorem}
\label{thm:planverif_chains_hard}
\pverify is \NP-complete, even when restricted to primitive task networks such that $\prec^+$ is a collection of chains, \(|A|=2\), and \(\pre(a)=\del(a)=\add(a)=\emptyset\) for all actions \(a\in A\).
\end{theorem}
\begin{proof}
    We reduce from \textsc{Shuffle Product} with the alphabet \(\Sigma = \set{a,b}\).
    Let \(A=\Sigma\).
    Construct a task for each letter in the words \(c_1,\ldots,c_w\); let the task \(t_{i,j}\) refer to the \(j^{\text{th}}\) letter of the word \(c_i\), \(i\in [w]\).
    Let \(\prec^+\) be such that the tasks for each word form a chain in the respective order, i.e., let \(\prec\) consist of the edges \((t_{i,j},t_{i,j+1})\), for all \(i\in [w]\) and \(j\in [|c_i|-1]\). 
    Then, \(\prec^+\) partitions \(T\) into \(w\) chains (LEFT side of \cref{fig:reduction_wordshuffle}). 
    Now, set \(\alpha(t_{i,j})\) to the action representing the \(j^{\text{th}}\) letter of \(c_i\).
    Let the input plan \(p\) be such that \(p=u\).
    The correctness of the reduction follows since verifying whether the tasks 
    can be arranged to match the input plan corresponds to interlacing the words to form \(u\). 
\end{proof}

\begin{theorem}
\label{thm:target_cover_chains_hard}
\ptarget and \pexist are \npcomplete, even when restricted to primitive task networks such that $\prec^+$ is a collection of chains, \(|A|\leq 5\), and the state transition graph has at most $4$ states.
\end{theorem}
\begin{proof}[Proof Sketch]
    We use a similar reduction as for \cref{thm:planverif_chains_hard}, (see \cref{fig:reduction_wordshuffle}), but encode $u$ as a separate chain and adapt the actions so that executable plans must alternate between LEFT and RIGHT tasks, and executed RIGHT tasks must be preceded by executed LEFT tasks of the same color.
\end{proof}
    \begin{figure}
        \centering
        \includegraphics[width=.32\textwidth]{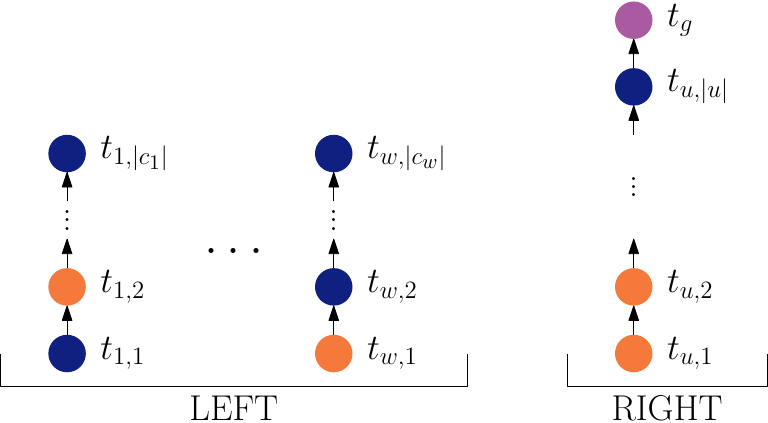}
        \caption{Constructed task network for reductions from \textsc{Shuffle Product} with words \(c_1,\ldots, c_w\) and \(u\).
        Each word is represented by a chain of tasks, where each task represents a letter in that word.
        The two letters in the alphabet are represented in orange and blue.}
        \label{fig:reduction_wordshuffle}     
    \end{figure}

\paragraph*{Chains and Antichains.}\label{sec:prim_chains_antichains}
Given the seemingly discouraging lower bounds obtained so far, one might wonder whether there are natural structural restrictions on the network which can be exploited to obtain polynomial-time algorithms. Here, we chart the route towards identifying these by first establishing the tractability of the considered problems (under the state-space restrictions justified in Sec.~\ref{sec:states}) in the extremes of~\(\prec^+\), where \(T\) is a chain (total order) or antichain (\(\prec^+\ =\emptyset\)).

When $T$ is a primitive chain, it is easy to observe that 
all of our problems are polynomial-time solvable (see, e.g.,~\cite{Erol_Hendler_Nau_1996,Behnke_Hoeller_Biundo_2015}).
\pverify restricted to primitive antichains is also trivial,
as it suffices to check whether the plan is executable in the initial state and the multiset of actions defined by the tasks in the network matches the actions in the plan. 
Handling antichains for the other problems requires more effort, and no tractability results exist that only assume the state space to be bounded. 

\begin{theorem}
\label{thm:reachtarget_fpt_empty}
\ptarget is polynomial-time solvable on primitive antichain task networks with constantly many states in the state transition graph.
\end{theorem}

\begin{proof}
    Let $k$ be the number of states, \(s_0\) the initial state, and \(s_g\) the target state. We first compute the state transition graph~$\G$ of $\I$ in quadratic time. 
    In \(\G\), if there are more than \(k\) edges between any pair of states, select any \(k\) of them and discard the rest. 
    Then, consider all simple edge paths from \(s_0\) to \(s_g\).
    Note that there are \bigO{2^k\cdot k^{k-1}} such paths, as there are \bigO{2^k} simple vertex paths, each of length at most \(k-1\), and each step might choose between up to \(k\) different edges.  
    For each simple edge path, check whether there is any action that is used more times in the path than there are tasks of that action in \(T\).
    If there exists a simple edge path for which the number of required tasks does not exceed the number of available ones, 
then give a positive answer, and otherwise reject.

    If such a path is found, the instance is indeed a yes-instance as there is a task sequence that results in a path to \(s_g\) in \(\G\).
    For the other direction, let there be a task sequence that leads to the state \(s_g\).
    Consider the subsequence obtained by exhaustively removing cycles in \(\G\) from that sequence.
    The resulting task sequence describes a simple edge path from \(s_0\) to \(s_g\). 
    It might employ edges that were discarded from \(\G\) by the algorithm.
    However, in this case there are \(k\) other edges connecting the respective states in \(\G\).
    Thus, one can choose to replace deleted edges by existing ones in such a way that each replacement action occurs at most once in the path, which ensures that there are enough tasks to actually walk this path.
    Hence, there is a path from \(s_0\) to \(s_g\) that only uses edges that are present in \(\G\) after the deletion, and for which the number of required tasks does not exceed the number of available ones. 
    As the algorithm will find this path, it will correctly give a positive output.
    Finally, the overall running time can be upper-bounded by \bigO{|\I|^2+|A|k+|T|+2^k\cdot k^k}.
\end{proof} 

To obtain an analogous result for \pcover, a more sophisticated approach is necessary, since a solution may be required to visit a state multiple times.

\begin{theorem}
\label{thm:includeaction_fpt_empty}
	\pcover is polynomial-time solvable on primitive antichain task networks with constantly many states in the state transition graph. 
\end{theorem}
\begin{proof}[Proof Sketch]
    \pcover on primitive antichains can be expressed by deciding whether there is a walk in the state transition graph such that each action is used in the path at least as many times as required by $S$ and at most as many times as there are respective tasks in $T$. 
    We bound the number of actions by a function of the number $k = O(1)$ of states by identifying \emph{equivalent} actions that have the same impact on the states in the state transition graph.  
    Now we enumerate all possible \emph{types} of walks in constant time by branching on a ``main-path'' from $s_0$ to any $s \supseteq s_g$ and which sets of simple cycles the walk uses. 
    For all branches where the main-path and all selected cycles are connected, it remains to determine the number of times each cycle is traversed.
    We do so by solving an \emph{Integer Linear Program} (ILP) where there are variables stating the number of traversals of each of the $O(1)$ selected cycles, and a constraint for each action requiring that the number of times it is executed (as defined by the variables and the number of times it occurs in the main-path and the respective cycles) is admissible w.r.t. $S$ and $T$.   
\end{proof}

\paragraph*{Generalized Partial Order Width.}\label{sec:prim_powidth}
We have seen that the considered problems are polynomial-time solvable on primitive chains and antichains. Here, we unify and extend these results by establishing tractability for primitive HTNs of constant generalized partial order width (chains and antichains have generalized partial order width $1$ and $0$, respectively). 

\begin{theorem}
\label{thm:planverif_powidth_xp}
	\pverify is polynomial-time solvable on primitive task networks of constant generalized partial order width.
\end{theorem}
\begin{proof}[Proof Sketch]
    We use dynamic programming to decide, for each index $i$ in the plan and all combinations of indices $h_1,\ldots,h_w$ in the $w=O(1)$ many chains partitioning the non-isolated tasks, whether there is a solution matching the first $i$ actions in the plan and using exactly the first $h_1,\ldots,h_w$ tasks of each chain, respectively. Clearly, this is possible for $i=h_1=\ldots = h_w = 0$ and we iteratively add a next admissable task from a chain or the isolated tasks.  
\end{proof}

\begin{theorem}
\label{thm:target_cover_powidth_xp}
	\ptarget and \pcover are polynomial-time solvable on primitive task networks with constantly many states in the state transition graph and constant generalized partial order width.
\end{theorem}
\begin{proof}[Proof Sketch] 
We note that the number $|E|$ of action equivalence classes that have the same impact on the state graph is constant. 
We use this in a dynamic programming approach similar to the one used for \cref{thm:planverif_powidth_xp}. 
Here, we decide for each state $s$, indices $h_j, j\in[w]$ in the $w$ chains, and numbers of actions $r_i, i\in [|E|]|$ in the equivalence classes, whether there is a solution with a linearization that
    uses exactly the first $h_j$ tasks of each chain $c_j$, 
    results in state $s$, 
    and uses exactly $r_i$ tasks of each action equivalence class $e_i$.
    We initialize all variables to be \False except for the one with $s=s_0$ and all $h_j$ and $r_i$ are $0$. When a variable is set to \True (including the initial one), for each next index in a chain and each equivalence class for which the $h_j$ and $r_i$ imply that there is still an unused isolated task, we test whether this task can be appended to a sequence witnessing the \True variable. If so, the corresponding variable is set to \True as well. 
\end{proof}

We recall that Theorem~\ref{thm:target_cover_powidth_xp} also immediately implies the analogous result for \pexist.

\section{Hierarchical Task Networks}\label{sec:htns}
Considering HTNs with compound tasks adds another, seemingly opaque layer of difficulty to the considered planning problems. In this section, we show that---at least in terms of classical complexity---we can cleanly characterize the jump from primitive to compound through three measures on compound task networks: we obtain an algorithmic meta-theorem that allows to lift polynomial-time solvability from primitive to compound task networks if all three of these measures are bounded. Moreover, this result is tight in the sense that none of the three measures can be dropped. The three measures we need are $\compn$, $\comps$, and $\compd$:

\begin{definition}
For an initial task network \mbox{\(tn=(T,\prec^+,\alpha)\)} in a domain $D$, we refer to the number of compound tasks in \(tn\) by 
\(\compn = |\set{t\in T\mid \alpha(t)\in \C}|\).
The maximum size of a network that a task can be decomposed into is
\(\comps = \max\set{|T_m|\mid (c, (T_m,\prec^+_m,\alpha_m))\in M}\).
The maximum depth \compd of a network is defined recursively:
\(\compd=0\) for primitive task networks, and a compound task network has \(\compd=i\) if $i$ is the smallest integer such that every decomposition of all compound tasks in the initial network $T$ results in a new network where \(\compd\leq i-1\). We let $\compd=\infty$ if no such $i$ exists (cyclic decomposition).
\end{definition}

For our later considerations, it will be useful to additionally define the maximum number of pairwise non-isomorphic networks that a task can be decomposed into; we denote that by \(\compc = \max_{c\in \C} |\set{tn_m \mid (c, tn_m)\in M}|\). The next technical lemma shows how these measures help us bound the ``decomposition complexity'' of compound HTNs.

\begin{lemma}
\label{lem:howmanydecompositions}
    If $\compd\neq \infty$, then a task network \(tn=(T,\prec^+,\alpha)\) can be decomposed into at most \(\compc^{\sum_{i=0}^{\compd-1} \compn\cdot \comps^i}\) pairwise non-isomorphic primitive task networks \(tn' = (T',{{\prec^+}}',\alpha')\).
    For each of these, \(|T'|\le |T|+\compn\cdot (\comps^{\compd}-1)\).
\end{lemma}
 \begin{proof}
    We first show the second part of the statement.
    The \(|T|-\compn\) primitive tasks in \(tn\) are still present in any network \(tn'\) that \(tn\) can be decomposed into.
    A compound task in \(tn\) can be decomposed into at most \comps compound tasks in one decomposition step and there are at most \(\compd\) decomposition steps. 
    Thus, each compound task in \(tn\) can create at most \(\comps^\compd\) tasks in \(tn'\).
    Hence, \(|T'|\le |T|-\compn +\compn\cdot \comps^\compd\).

    For the first statement, note that decomposing each of the \(\compn\) compound tasks once in a task network of depth \(\compd\) can result in at most \(\compc^\compn\) pairwise non-isomorphic task networks of depth at most \(\compd-1\).
	Further, each of these resulting networks has at most \(\compn\cdot \comps\) compound tasks.
    Thus, \(tn\) can be decomposed into at most
    \begin{align*}
        \compc^{\compn} \cdot \compc^{\compn\cdot \comps} \cdot \ldots \cdot \compc^{\compn\cdot \comps^{\compd-1}}
        = \compc^{\sum_{i=0}^{\compd-1} \compn\cdot \comps^i}
    \end{align*}
    pairwise non-isomorphic task networks.    
\end{proof} 

Having defined our measures of interest, we proceed to formalize the set of problems our meta-theorem will capture. Intuitively, the aim is to formalize the set of all problems which deal with compound HTNs by decomposing them.

\begin{definition}
    Let \textsc{PR} be an arbitrary HTN problem.
    Then, \textsc{PR} is {\it decomposable} if and only if, for all non-primitive initial task networks \(tn\), \textsc{PR} on \(tn\) is a yes-instance if and only if there is a primitive task network \(tn'\) such that \(tn\) can be decomposed into \(tn'\) and \(tn'\) is a yes-instance of \textsc{PR}.
\end{definition}

Note that all problems discussed in this work are decomposable.
Next, in order to capture the general notions of ``tractability'' for various HTN problems on networks where certain measures are bounded, we need to slightly restrict the measures in question to avoid entirely degenerate measures.

\begin{definition}
    Let \(\kappa\) be any numerical measure of HTNs in a domain $D$, i.e., a mapping that assigns each HTN a non-negative number. We say that $\kappa$ is \emph{stable} if there exists a computable function~$f$ with the following property: for each primitive HTN~$tn^*$ consisting of $c$ tasks and each HTN~$tn$ containing some compound task $t$, the HTN~$tn'$ obtained from $tn$ by decomposing $t$ into $tn^*$ satisfies $\kappa(tn', D)\leq f(\kappa(tn, D),c)$.
\end{definition}

We say that a stable measure $\kappa$ is a \emph{stable tractability measure} for a decomposable HTN problem \textsc{PR} if, for each constant $c$, \textsc{PR} can be solved in polynomial time on all primitive HTNs $tn$ such that $\kappa(tn, D)\leq c$.
Notably, the (generalized) partial order width measure fits in that definition, where its stability is due to the fact that the addition of $c$ vertices to a graph can only increase its (generalized) partial order width by at most $c$. 
\begin{observation}
\label{obs:postable}
    (Generalized) partial order width is a stable tractability measure for \pverify. 
    The sum of the (generalized) partial order width and the number of states in the state space is a stable tractability measure for \pcover and \ptarget.
\end{observation} 

We are now ready to establish the first meta-theorem.

\begin{theorem}
\label{thm:meta_for_compounds}
    Let {\sc PR} be a decomposable HTN problem and $\kappa$ a stable tractability measure. 
    Then, {\sc PR} is polynomial-time solvable when restricted to any initial network $tn$ in domain~$D$ where
    $\compdns,$ and $\kappa(tn,D)$ are constants.
\end{theorem}
\begin{proof}
    Let \(\I\) have task network \(tn=(T,\prec^+,\alpha)\).
    If there are no compound tasks in \(T\), then we apply an algorithm for primitive networks.
    Thus, assume \(\compd\ge 1\). 
    Then, \textsc{PR} in \(tn\) can be solved by listing all the primitive instances \(tn\) can be decomposed into and solving \(\textsc{PR}\) in each of these.
    By Lemma~\ref{lem:howmanydecompositions}, there are at most \(\compc^{g(\compdns)} \leq |\I|^{g(\compdns)}\) such instances \(\I_p\) for \(g(\compdns)=\sum_{i=0}^{\compd-1} \compn\cdot \comps^i\) and, for each of these, \(|\I_p|\le |\I| + f''(\compdns)\) for some computable function \(f''\).
    Furthermore, for each initial task network \(tn_p\) of an instance \(\I_p\), we have \(\kappa(tn_p, D) \le f'''(\kappa(tn, D),\compdns)\), for some computable function \(f'''\) as \(\kappa\) is a stable measure and $tn_p$ is obtained from $tn$ by sequentially decomposing $\compn$ compound tasks into a total of $f''(\compdns)$ primitive tasks.
    Thus, if primitive instances \(\I_p\) can be solved in time at most \(|\I_p|^{f(\kappa(tn_p, D))}\), then the total runtime (for some computable function $f'$) is 
    \begin{equation*}
    \begin{split}
         \mathcal{O}\Big((|\I| + f''(\compdns))^{f(f'''(\kappa(tn, D),\compdns))}\\
         \cdot \compc^{g(\compdns)}\Big)   
         = \bigO{|\I|^{f'(\kappa(tn, D),\compdns)}}.\qedhere
         \end{split}
    \end{equation*}
\end{proof}

We complement Theorem~\ref{thm:meta_for_compounds} via lower bounds which show that bounding \compd, \compn, and \comps is necessary.
Specifically, we prove that if any one of these measures is not bounded by a constant, then some of the considered problems become \nphard\ on compound networks, even for settings that were easily solvable on primitive networks. 
While other reductions show the hardness deriving from decomposing non-primitive networks (see, e.g.,~\cite{Erol_Hendler_Nau_1996,Alford_Bercher_Aha_2015,Behnke_Hoeller_Biundo_2015}), our reduction is stronger in the sense that it bounds all but any one of the parameters \compd, \compn, \comps. 

\begin{theorem}
\label{thm:meta_theorem_lowerbound}
For each $\iota\in \{\compdns\}$, \pexist, \pcover, and \ptarget remain \nphard\ when restricted to instances which have (generalized) partial order width at most $1$ and where $\{\compdns, \compc\}\setminus \{\iota\}$ are upper-bounded by a constant.
\end{theorem}
\begin{proof}[Proof Sketch]
We reduce from \textsc{Multicolored-Clique}, which is \nphard and asks whether a given properly $k$-vertex-colored graph \(G\) contains a $k$-clique~\cite{FELLOWS2009}. The reduction creates a compound task for each vertex and edge of $G$. These can all either be decomposed into a primitive vertex- or edge-task, or an empty ``dummy'' network. The actions and state space of the output instance are set up in a way which ensures that, for each color class, only a single primitive vertex task can be executed. The vertices for which we execute the corresponding primitive vertex task are considered as ``chosen'', and a primitive edge task can only be executed if both of its incident vertices are chosen. Based on the targeted HTN problem, we ensure the instance admits a solution if and only if it is possible to execute $k \choose 2$ primitive edge tasks, matching each of the $k \choose 2$ pairs of colors. 

The above completes the reduction for the case where $\iota=\compn$. For the case where $\iota=\comps$, we slightly adapt the construction by having the whole network (as described above) be the result of decomposing a single initial compound task $t_0$. Finally, when $\iota=\compd$ we build on the case of $\iota=\comps$ by making the decomposition from $t_0$ ``gradual'' as opposed to only occurring in a single step.
\end{proof}

\section{A Parameterized Analysis of HTNs}
\label{sec:parcomp}

In this final section, we ask whether---and for which parameterizations---our problems on primitive HTNs admit fixed-parameter algorithms. 
While there are many graph-theoretic parameters typically used to achieve fixed-parameter tractability, none of the ``usual suspects'' in that regard will help with our problems due to the hardness on forests of stars \cite{Lin_Bercher_2023}  
and sets of chains (\cref{thm:planverif_chains_hard,thm:target_cover_chains_hard}).
Indeed, this applies to \emph{treewidth}~\cite{RobertsonS86}, the \emph{feedback edge number}~\cite{GanianK21}, and similar measures defined on digraphs~\cite{GanianHK0ORS16}.

On the other hand, our problems are polynomial-time solvable for constant generalized partial order width and so we need a more refined notion of hardness to rule out fixed-parameter tractability in this case.
In particular, hardness w.r.t.\ a parameter $k$ for the complexity classes $\W[1]$ or $\W[2]$ rules out any fixed-parameter algorithm under the well-established assumption that $\W[1]\neq \FPT$.
As the \textsc{Shuffle Product} problem is \W[2]-hard when parameterized by the width~\(w\)~\cite{Bevern_etal_2016}, 
\cref{thm:planverif_chains_hard,thm:target_cover_chains_hard} also imply the \W[2]-hardness of the considered problems.

\paragraph*{Fixed-Parameter Tractability via Vertex Cover.}
We contrast these lower bounds with a fixed-parameter algorithm for primitive networks that exploits a different parameterization, namely the vertex cover number (vcn) of \(G_\prec\).
A vertex cover of \(G_\prec\) is a subset of tasks \(V\subseteq T\) such that, for all \(t \prec t'\), \(t\in V\) and/or \(t'\in V\); the vcn of \(G_\prec\) is the minimum size of a vertex cover of \(G_\prec\). 
While the vcn has been used to obtain fixed-parameter algorithms for many other problems~\cite{BalkoCG00V022,BlazejGKPSS23,Icalpvc,FionnColt}, the techniques from previous works do not easily translate to HTN problems. 
Instead, our algorithms use a delicate branching routine whose correctness requires a careful analysis of the state space.

\begin{theorem}
\label{thm:planverif_vertexcover}
    \pverify is fixed-parameter tractable on primitive task networks parameterized by the vcn of $G_\prec$.
\end{theorem}
\begin{proof}[Proof Sketch]
    We must order the tasks so that they respect~$\prec$ and correspond to the input plan $p$.
    We branch over all orderings of the tasks in a minimal vertex cover \(V\) and discard those which do not respect $\prec$. 
    Let the \emph{priority} of a task be higher the earlier one of its out-neighbors in $V$ is placed in the ordering. Tasks in $V$ have priority just below the last of their in-neighbors, and tasks outside $V$ with no out-edges have no priority.
    We iterate through $p$ and, at each step, select a task for that position that has the respective action and has all of its predecessors in $\prec$ already selected. 
    If there are multiple candidates, we choose a task with highest priority, ensuring that order constraints are met as early as possible.    
\end{proof}

\begin{theorem}
\label{thm:reachtarget_includeaction_fpt}
    \ptarget and \pcover on primitive networks are fixed-parameter tractable when parameterized by the vcn of $G_\prec$ plus the number of states in the state transition graph.
\end{theorem}
\begin{proof}[Proof Sketch]
    We branch over the order of vertices in the vertex cover as in \cref{thm:planverif_vertexcover}, but need more involved arguments for constructing the task sequence. Inspired by the proof of \cref{thm:includeaction_fpt_empty}, we carefully branch on which paths and cycles the sequence takes in the state transition graph in between visiting the tasks of the vertex cover, and use an ILP to determine the number of times each cycle is traversed. 
\end{proof}

\paragraph*{A Meta-Theorem for Fixed-Parameter Tractability.}

As our final contribution, we show that the algorithmic meta-theorem for polynomial-time solvability can be lifted to the setting of fixed-parameter tractability if we also restrict the maximum ``breadth'' of a compound task, i.e., $\compc$.
To do this, we first need to define a suitable notion of stability. We call~$\kappa$ a \emph{stable fixed-parameter tractability measure} for a decomposable HTN problem \textsc{PR} if $\kappa$ is a stable measure and \textsc{PR} is fixed-parameter tractable on primitive task networks parameterized by $\kappa$. The vcn is an example of this notion as the addition of $c$ vertices to a graph increases its vcn by at most~$c$.

\begin{observation}
    The vcn is a stable fixed-parameter tractability measure for \pverify. 
    The sum of the  vcn and the number of states in the state space is a stable fixed-parameter tractability measure for \pcover and \ptarget.
\end{observation} 

We now establish the parameterized analog of \cref{thm:meta_for_compounds}, which extends \cref{thm:planverif_vertexcover,thm:reachtarget_includeaction_fpt} to compound networks:

\begin{theorem}
\label{thm:meta_for_compounds_fpt}    
    Let {\sc PR} be a decomposable HTN problem and $\kappa$ a stable fixed-parameter tractability measure. 
    If $\compd \neq \infty$, then {\sc PR} is fixed-parameter tractable by the combined parameters $\compcdns,$ and $\kappa(tn,D)$ for the initial network $tn$. 
\end{theorem}
\begin{proof}
    Let \(tn=(T,\prec^+,\alpha)\).
    If there are no compound tasks in \(T\), i.e., \(\compd=0\), the statement immediately holds by applying an algorithm for primitive networks.
    Thus, assume \(\compd\ge 1\). 
    Then, \textsc{PR} in \(tn\) is solved by listing all primitive instances \(tn\) can be decomposed into and solving \(\textsc{PR}\) in each of these.
    By Lemma~\ref{lem:howmanydecompositions}, there are at most \(\compc^{\sum_{i=0}^{\compd-1} \compn\cdot \comps^i}\) such instances \(\I_p\) and, for each of these, \(|\I_p|\le |\I| + f''(\compdns)\) for some computable function \(f''\).
    Further, for each initial task network \(tn_p\) of an instance \(\I_p\), we have \(\kappa(tn_p, D) \le f'''(\kappa(tn, D),\compdns)\), for some computable function \(f'''\) as \(\kappa\) is a stable measure.
    Thus, if primitive instances \(\I_p\) can be solved in time \(f(\kappa(tn_p,D)) |\I_p|^\bigO{1}\), then the total running time can be upper-bounded by $\bigO{\compc^{g(\compdns)} \cdot f'(\kappa(tn,D),\compdns)\cdot |\I|^{\bigO{1}}}$ for computable functions $g$ and $f'$. 
\end{proof}

\section{Concluding Remarks}
This article provides a comprehensive understanding of the complexity-theoretic landscape of several fundamental problems on HTNs. Our results include strong algorithmic lower bounds and complementary positive results---not only for specific problems, but also in the form of meta-theorems that can be used for other HTN problems. 
For further work, while we have provided lower bounds justifying all the restrictions applied in the algorithmic meta-theorem for polynomial-time solvability (\cref{thm:meta_for_compounds}), in the more refined fixed-parameter tractability setting we leave it open whether the ``hierarchical breadth'' $\compc$ needs to be part of the parameterization in order to establish \cref{thm:meta_for_compounds_fpt}.

\subsubsection{Acknowledgements}
This research was funded in whole or in part by the Austrian Science Fund (FWF) [10.55776/Y1329].

\bibliographystyle{named}
\bibliography{ref_abbrev}

\begin{thebibliography}{}

\bibitem[\protect\citeauthoryear{Alford \bgroup \em et al.\egroup
  }{2014}]{Alford_Shivashankar_Kuter_Nau_2014}
Ron Alford, Vikas Shivashankar, Ugur Kuter, and Dana Nau.
\newblock On the feasibility of planning graph style heuristics for {HTN}
  planning.
\newblock In {\em ICAPS}, 2014.

\bibitem[\protect\citeauthoryear{Alford \bgroup \em et al.\egroup
  }{2015a}]{Alford_Bercher_Aha_2015}
Ron Alford, Pascal Bercher, and David Aha.
\newblock Tight bounds for htn planning.
\newblock In {\em ICAPS}, 2015.

\bibitem[\protect\citeauthoryear{Alford \bgroup \em et al.\egroup
  }{2015b}]{Alford2015TightBF}
Ron Alford, Pascal Bercher, and David~W. Aha.
\newblock Tight bounds for {HTN} planning with task insertion.
\newblock In {\em IJCAI}, 2015.

\bibitem[\protect\citeauthoryear{Alford \bgroup \em et al.\egroup
  }{2016}]{AlfordBHBBA16}
Ron Alford, Gregor Behnke, Daniel H{\"{o}}ller, Pascal Bercher, Susanne Biundo,
  and David~W. Aha.
\newblock Bound to plan: Exploiting classical heuristics via automatic
  translations of tail-recursive {HTN} problems.
\newblock In {\em ICAPS}, 2016.

\bibitem[\protect\citeauthoryear{B{\"{a}}ckstr{\"{o}}m and
  Nebel}{1995}]{Backstrom_Nebel_1995}
Christer B{\"{a}}ckstr{\"{o}}m and Bernhard Nebel.
\newblock Complexity results for {SAS}$^+$ planning.
\newblock {\em Computational Intelligence}, 11(4):625--655, 1995.

\bibitem[\protect\citeauthoryear{B{\"{a}}ckstr{\"{o}}m \bgroup \em et
  al.\egroup }{2012}]{BackstromCJOS12}
Christer B{\"{a}}ckstr{\"{o}}m, Yue Chen, Peter Jonsson, Sebastian Ordyniak,
  and Stefan Szeider.
\newblock The complexity of planning revisited - {A} parameterized analysis.
\newblock In {\em AAAI}, 2012.

\bibitem[\protect\citeauthoryear{B{\"{a}}ckstr{\"{o}}m \bgroup \em et
  al.\egroup }{2013}]{BackstromJOS13}
Christer B{\"{a}}ckstr{\"{o}}m, Peter Jonsson, Sebastian Ordyniak, and Stefan
  Szeider.
\newblock Parameterized complexity and kernel bounds for hard planning
  problems.
\newblock In {\em CIAC}, 2013.

\bibitem[\protect\citeauthoryear{Balko \bgroup \em et al.\egroup
  }{2022}]{BalkoCG00V022}
Martin Balko, Steven Chaplick, Robert Ganian, Siddharth Gupta, Michael
  Hoffmann, Pavel Valtr, and Alexander Wolff.
\newblock Bounding and computing obstacle numbers of graphs.
\newblock In {\em ESA}, 2022.

\bibitem[\protect\citeauthoryear{Behnke and Speck}{2021}]{Behnke_Speck_2021}
Gregor Behnke and David Speck.
\newblock Symbolic search for optimal total-order {HTN} planning.
\newblock In {\em AAAI}, 2021.

\bibitem[\protect\citeauthoryear{Behnke \bgroup \em et al.\egroup
  }{2015}]{Behnke_Hoeller_Biundo_2015}
Gregor Behnke, Daniel Höller, and Susanne Biundo.
\newblock On the complexity of {HTN} plan verification and its implications for
  plan recognition.
\newblock In {\em ICAPS}, 2015.

\bibitem[\protect\citeauthoryear{Behnke \bgroup \em et al.\egroup
  }{2019}]{BehnkeHB19}
Gregor Behnke, Daniel H{\"{o}}ller, and Susanne Biundo.
\newblock Finding optimal solutions in {HTN} planning - {A} {SAT}-based
  approach.
\newblock In {\em IJCAI}, 2019.

\bibitem[\protect\citeauthoryear{Bercher \bgroup \em et al.\egroup
  }{2019}]{Bercher_Alford_Hoeller_2019}
Pascal Bercher, Ron Alford, and Daniel Höller.
\newblock A survey on hierarchical planning – one abstract idea, many
  concrete realizations.
\newblock In {\em IJCAI}, 2019.

\bibitem[\protect\citeauthoryear{Blazej \bgroup \em et al.\egroup
  }{2023}]{BlazejGKPSS23}
V{\'{a}}clav Blazej, Robert Ganian, Dusan Knop, Jan Pokorn{\'{y}}, Simon
  Schierreich, and Kirill Simonov.
\newblock The parameterized complexity of network microaggregation.
\newblock In {\em AAAI}, 2023.

\bibitem[\protect\citeauthoryear{Bodlaender \bgroup \em et al.\egroup
  }{2023}]{Icalpvc}
Hans~L. Bodlaender, Carla Groenland, and Michal Pilipczuk.
\newblock Parameterized complexity of binary {CSP:} vertex cover, treedepth,
  and related parameters.
\newblock In {\em ICALP}, 2023.

\bibitem[\protect\citeauthoryear{Bylander}{1994}]{Bylander94}
Tom Bylander.
\newblock The computational complexity of propositional {STRIPS} planning.
\newblock {\em Artif. Intel.}, 69(1-2):165--204, 1994.

\bibitem[\protect\citeauthoryear{Bäckström \bgroup \em et al.\egroup
  }{2015}]{Backstrom_Jonsson_Ordyniak_Szeider_2015}
Christer Bäckström, Peter Jonsson, Sebastian Ordyniak, and Stefan Szeider.
\newblock A complete parameterized complexity analysis of bounded planning.
\newblock {\em J. Comput. Syst. Sci.}, 81(7):1311--1332, 2015.

\bibitem[\protect\citeauthoryear{Chalopin \bgroup \em et al.\egroup
  }{2024}]{FionnColt}
J{\'{e}}r{\'{e}}mie Chalopin, Victor Chepoi, Fionn {Mc~Inerney}, and
  S{\'{e}}bastien Ratel.
\newblock Non-clashing teaching maps for balls in graphs.
\newblock In {\em COLT}, 2024.

\bibitem[\protect\citeauthoryear{Chapman}{1987}]{Chapman_1987}
David Chapman.
\newblock Planning for conjunctive goals.
\newblock {\em Artif. Intel.}, 32(3):333--377, 1987.

\bibitem[\protect\citeauthoryear{Chen and Bercher}{2021}]{Chen_Bercher_2021}
Dillon Chen and Pascal Bercher.
\newblock Fully observable nondeterministic {HTN} planning -- {F}ormalisation
  and complexity results.
\newblock In {\em ICAPS}, 2021.

\bibitem[\protect\citeauthoryear{Cygan \bgroup \em et al.\egroup
  }{2015}]{CyganFKLMPPS15}
Marek Cygan, Fedor~V. Fomin, {\L{}}ukasz Kowalik, Daniel Lokshtanov,
  D{\'{a}}niel Marx, Marcin Pilipczuk, Micha{\l{}} Pilipczuk, and Saket
  Saurabh.
\newblock {\em Parameterized Algorithms}.
\newblock Springer, Cham, 2015.

\bibitem[\protect\citeauthoryear{Diestel}{2012}]{Diestel}
Reinhard Diestel.
\newblock {\em Graph Theory, 4th Edition}, volume 173 of {\em Graduate texts in
  math.}
\newblock Springer, 2012.

\bibitem[\protect\citeauthoryear{Downey and Fellows}{2013}]{DowneyFellows13}
Rodney~G. Downey and Michael~R. Fellows.
\newblock {\em Fundamentals of Parameterized Complexity}.
\newblock Texts in Computer Science. Springer, 2013.

\bibitem[\protect\citeauthoryear{Eiben \bgroup \em et al.\egroup
  }{2021}]{EibenGKOS21}
Eduard Eiben, Robert Ganian, Iyad Kanj, Sebastian Ordyniak, and Stefan Szeider.
\newblock The parameterized complexity of clustering incomplete data.
\newblock In {\em AAAI}, 2021.

\bibitem[\protect\citeauthoryear{Erol \bgroup \em et al.\egroup
  }{1994}]{ErolHN94}
Kutluhan Erol, James~A. Hendler, and Dana~S. Nau.
\newblock {HTN} planning: Complexity and expressivity.
\newblock In {\em AAAI}, 1994.

\bibitem[\protect\citeauthoryear{Erol \bgroup \em et al.\egroup
  }{1996}]{Erol_Hendler_Nau_1996}
Kutluhan Erol, James Hendler, and Dana Nau.
\newblock Complexity results for {HTN} planning.
\newblock {\em Annals of Mathematics and Artificial Intelligence}, 18:69--93,
  1996.

\bibitem[\protect\citeauthoryear{Fellows and
  McCartin}{2003}]{Fellows_McCartin_2003}
Michael~R. Fellows and Catherine McCartin.
\newblock On the parametric complexity of schedules to minimize tardy tasks.
\newblock {\em Theor. Comput. Sci.}, 298(2):317--324, 2003.

\bibitem[\protect\citeauthoryear{Fellows \bgroup \em et al.\egroup
  }{2009}]{FELLOWS2009}
Michael~R. Fellows, Danny Hermelin, Frances Rosamond, and Stéphane Vialette.
\newblock On the parameterized complexity of multiple-interval graph problems.
\newblock {\em Theor. Comput. Sci.}, 410(1):53--61, 2009.

\bibitem[\protect\citeauthoryear{Fikes and Nilsson}{1971}]{Fikes_Nilsson_1971}
Richard~E. Fikes and Nils~J. Nilsson.
\newblock Strips: A new approach to the application of theorem proving to
  problem solving.
\newblock {\em Artif. Intel.}, 2(3):189--208, 1971.

\bibitem[\protect\citeauthoryear{Frank and Tardos}{1987}]{Frank_Tardos_1987}
András Frank and Éva Tardos.
\newblock An application of simultaneous diophantine approximation in
  combinatorial optimization.
\newblock {\em Combinatorica}, 7(1):49–65, 1987.

\bibitem[\protect\citeauthoryear{Froese \bgroup \em et al.\egroup
  }{2022}]{FroeseKN22}
Vincent Froese, Leon Kellerhals, and Rolf Niedermeier.
\newblock Modification-fair cluster editing.
\newblock In {\em AAAI}, 2022.

\bibitem[\protect\citeauthoryear{Ganian and Korchemna}{2021}]{GanianK21}
Robert Ganian and Viktoriia Korchemna.
\newblock The complexity of {B}ayesian network learning: Revisiting the
  superstructure.
\newblock In {\em NeurIPS}, 2021.

\bibitem[\protect\citeauthoryear{Ganian \bgroup \em et al.\egroup
  }{2016}]{GanianHK0ORS16}
Robert Ganian, Petr Hlinen{\'{y}}, Joachim Kneis, Daniel Meister, Jan
  Obdrz{\'{a}}lek, Peter Rossmanith, and Somnath Sikdar.
\newblock Are there any good digraph width measures?
\newblock {\em J. Comb. Theory, Ser. {B}}, 116, 2016.

\bibitem[\protect\citeauthoryear{Geier and Bercher}{2011}]{Geier_Bercher_2011}
Thomas Geier and Pascal Bercher.
\newblock On the decidability of {HTN} planning with task insertion.
\newblock In {\em IJCAI}, 2011.

\bibitem[\protect\citeauthoryear{Georgievski and
  Aiello}{2015}]{Georgievski_Aiello_2015}
Ilche Georgievski and Marco Aiello.
\newblock {HTN} planning: Overview, comparison, and beyond.
\newblock {\em Artif. Intel.}, 222:124--156, 2015.

\bibitem[\protect\citeauthoryear{Gonz{\'{a}}lez{-}Ferrer \bgroup \em et
  al.\egroup }{2013a}]{Gonzalez-FerrerFC13}
Arturo Gonz{\'{a}}lez{-}Ferrer, Juan Fern{\'{a}}ndez{-}Olivares, and Luis~A.
  Castillo.
\newblock From business process models to hierarchical task network planning
  domains.
\newblock {\em Knowl. Eng. Rev.}, 28(2):175--193, 2013.

\bibitem[\protect\citeauthoryear{Gonz{\'{a}}lez{-}Ferrer \bgroup \em et
  al.\egroup }{2013b}]{Gonzalez-FerrerTFM13}
Arturo Gonz{\'{a}}lez{-}Ferrer, Annette ten Teije, Juan
  Fern{\'{a}}ndez{-}Olivares, and Krystyna Milian.
\newblock Automated generation of patient-tailored electronic care pathways by
  translating computer-interpretable guidelines into hierarchical task
  networks.
\newblock {\em Artif. Intel. Medicine}, 57(2):91--109, 2013.

\bibitem[\protect\citeauthoryear{Hansen \bgroup \em et al.\egroup
  }{2017}]{Hansen_Kaplan_Tarjan_Zwick_2017}
Thomas~Dueholm Hansen, Haim Kaplan, Robert~E. Tarjan, and Uri Zwick.
\newblock Hollow heaps.
\newblock {\em {ACM} Trans. Algorithms}, 13(3):42:1--42:27, 2017.

\bibitem[\protect\citeauthoryear{Hayes and Scassellati}{2016}]{HayesS16}
Bradley Hayes and Brian Scassellati.
\newblock Autonomously constructing hierarchical task networks for planning and
  human-robot collaboration.
\newblock In {\em ICRA}, 2016.

\bibitem[\protect\citeauthoryear{Heeger \bgroup \em et al.\egroup
  }{2023}]{HeegerHMMNS21}
Klaus Heeger, Danny Hermelin, George~B. Mertzios, Hendrik Molter, Rolf
  Niedermeier, and Dvir Shabtay.
\newblock Equitable scheduling on a single machine.
\newblock {\em J. Sched.}, 26:209--225, 2023.

\bibitem[\protect\citeauthoryear{H{\"{o}}ller \bgroup \em et al.\egroup
  }{2019}]{HollerBBB19}
Daniel H{\"{o}}ller, Pascal Bercher, Gregor Behnke, and Susanne Biundo.
\newblock On guiding search in {HTN} planning with classical planning
  heuristics.
\newblock In {\em IJCAI}, 2019.

\bibitem[\protect\citeauthoryear{Höller and
  Bercher}{2021}]{Hoeller_Bercher_2021}
Daniel Höller and Pascal Bercher.
\newblock Landmark generation in {HTN} planning.
\newblock In {\em AAAI}, 2021.

\bibitem[\protect\citeauthoryear{Höller \bgroup \em et al.\egroup
  }{2020}]{Hoeller_Bercher_Behnke_2020}
Daniel Höller, Pascal Bercher, and Gregor Behnke.
\newblock Delete- and ordering-relaxation heuristics for htn planning.
\newblock In {\em IJCAI}, 2020.

\bibitem[\protect\citeauthoryear{Kannan}{1987}]{Kannan_1987}
Ravi Kannan.
\newblock Minkowski’s convex body theorem and integer programming.
\newblock {\em Math. of Operations Research}, 12(3):415--–440, 1987.

\bibitem[\protect\citeauthoryear{Kronegger \bgroup \em et al.\egroup
  }{2013}]{Kronegger_Pfandler_Pichler_2013}
Martin Kronegger, Andreas Pfandler, and Reinhard Pichler.
\newblock Parameterized complexity of optimal planning: {A} detailed map.
\newblock In {\em IJCAI}, 2013.

\bibitem[\protect\citeauthoryear{Lenstra}{1983}]{Lenstra_1983}
H.~W. Lenstra.
\newblock Integer programming with a fixed number of variables.
\newblock {\em Math. of Operations Research}, 8(4):538–--548, 1983.

\bibitem[\protect\citeauthoryear{Li \bgroup \em et al.\egroup }{2009}]{LiKY09}
Nan Li, Subbarao Kambhampati, and Sung~Wook Yoon.
\newblock Learning probabilistic hierarchical task networks to capture user
  preferences.
\newblock In {\em IJCAI}, 2009.

\bibitem[\protect\citeauthoryear{Lin and Bercher}{2023}]{Lin_Bercher_2023}
Songtuan Lin and Pascal Bercher.
\newblock Was fixing this really that hard? {O}n the complexity of correcting
  {HTN} domains.
\newblock In {\em AAAI}, 2023.

\bibitem[\protect\citeauthoryear{Lin \bgroup \em et al.\egroup
  }{2024a}]{LinHB24}
Songtuan Lin, Daniel H{\"{o}}ller, and Pascal Bercher.
\newblock Modeling assistance for hierarchical planning: An approach for
  correcting hierarchical domains with missing actions.
\newblock In {\em SOCS}, 2024.

\bibitem[\protect\citeauthoryear{Lin \bgroup \em et al.\egroup
  }{2024b}]{LinOHB24}
Songtuan Lin, Conny Olz, Malte Helmert, and Pascal Bercher.
\newblock On the computational complexity of plan verification, (bounded)
  plan-optimality verification, and bounded plan existence.
\newblock In {\em AAAI}, 2024.

\bibitem[\protect\citeauthoryear{Liu \bgroup \em et al.\egroup
  }{2016}]{LiuWQZW16}
Dian Liu, Hongwei Wang, Chao Qi, Peng Zhao, and Jian Wang.
\newblock Hierarchical task network-based emergency task planning with
  incomplete information, concurrency and uncertain duration.
\newblock {\em Knowl. Based Syst.}, 112:67--79, 2016.

\bibitem[\protect\citeauthoryear{Nebel and
  Bäckström}{1994}]{Nebel_Backstrom_1994}
Bernhard Nebel and Christer Bäckström.
\newblock On the computational complexity of temporal projection, planning, and
  plan validation.
\newblock {\em Artificial Intelligence}, 66(1):125--160, 1994.

\bibitem[\protect\citeauthoryear{Olz \bgroup \em et al.\egroup
  }{2021}]{Olz_Biundo_Bercher_2021}
Conny Olz, Susanne Biundo, and Pascal Bercher.
\newblock Revealing hidden preconditions and effects of compound {HTN} planning
  tasks – {A} complexity analysis.
\newblock In {\em AAAI}, 2021.

\bibitem[\protect\citeauthoryear{Padia \bgroup \em et al.\egroup
  }{2019}]{PadiaBH19}
Kalpesh Padia, Kaveen~Herath Bandara, and Christopher~G. Healey.
\newblock A system for generating storyline visualizations using hierarchical
  task network planning.
\newblock {\em Comput. Graph.}, 78:64--75, 2019.

\bibitem[\protect\citeauthoryear{Robertson and Seymour}{1986}]{RobertsonS86}
Neil Robertson and Paul~D. Seymour.
\newblock Graph minors. {II.} {A}lgorithmic aspects of tree-width.
\newblock {\em J. Algorithms}, 7(3):309--322, 1986.

\bibitem[\protect\citeauthoryear{Sohrabi \bgroup \em et al.\egroup
  }{2009}]{SohrabiBM09}
Shirin Sohrabi, Jorge~A. Baier, and Sheila~A. McIlraith.
\newblock {HTN} planning with preferences.
\newblock In {\em IJCAI}, 2009.

\bibitem[\protect\citeauthoryear{Tsuneto \bgroup \em et al.\egroup
  }{1996}]{TsunetoEHN96}
Reiko Tsuneto, Kutluhan Erol, James~A. Hendler, and Dana~S. Nau.
\newblock Commitment strategies in hierarchical task network planning.
\newblock In {\em AAAI}, 1996.

\bibitem[\protect\citeauthoryear{van Bevern \bgroup \em et al.\egroup
  }{2016}]{Bevern_etal_2016}
René van Bevern, Robert Bredereck, Laurent Bulteau, Christian Komusiewicz,
  Nimrod Talmon, and Gerhard~J. Woeginger.
\newblock Precedence-constrained scheduling problems parameterized by partial
  order width.
\newblock In {\em DOOR}, 2016.

\bibitem[\protect\citeauthoryear{Warmuth and
  Haussler}{1984}]{Warmuth_Haussler_1984}
Manfred~K. Warmuth and David Haussler.
\newblock On the complexity of iterated shuffle.
\newblock {\em J. Comput. Syst. Sci.}, 28(3):345--358, 1984.

\bibitem[\protect\citeauthoryear{Xiao \bgroup \em et al.\egroup
  }{2017}]{XiaoHPWS17}
Zhanhao Xiao, Andreas Herzig, Laurent Perrussel, Hai Wan, and Xiaoheng Su.
\newblock Hierarchical task network planning with task insertion and state
  constraints.
\newblock In {\em IJCAI}, 2017.

\bibitem[\protect\citeauthoryear{Zhao \bgroup \em et al.\egroup
  }{2017}]{ZhaoQL17}
Peng Zhao, Chao Qi, and Dian Liu.
\newblock Resource-constrained hierarchical task network planning under
  uncontrollable durations for emergency decision-making.
\newblock {\em J. Intel. Fuzzy Syst.}, 33(6):3819--3834, 2017.

\end{thebibliography}

\iflong
\newpage
\appendix
\section{Appendix: Proof Details}

\begin{proof}[\textbf{Proof of \cref{thm:target_cover_chains_hard}}]
    We reduce from \textsc{Shuffle Product} with the alphabet \(\Sigma = \set{a,b}\).
    Construct a task for each letter in the words \(u,c_1,\ldots,c_w\), i.e., for each word \(x\in \set{u,c_1,c_2,\ldots,c_w}\), let the task \(t_{x,j}\) refer to its \(j^{\text{th}}\) letter.
    Let \(T_c\) be the set of tasks created for letters in \(c_1,\ldots,c_w\), and \(T_u\) the set of tasks \(t_{u,j}\), \(j\in|u|\).
    Let \(\prec^+\) be such that the tasks for each word form a chain in the respective order, i.e., let \(\prec\) consist of the edges \((t_{x,j},t_{x,j+1})\) for each \(x\in\set{u,c_1,\ldots,c_w}\) and \(j\in [|x|-1]\). 
    Further, add a task \(t_g\) and let \((t_{u,|u|}, t_g)\in \prec\).
    Then, \(\prec^+\) partitions \(T\) into \(w+1\) chains (see \cref{fig:reduction_wordshuffle}).

    For the sake of visualization and the corresponding naming of actions and propositional states, we consider the chains for \(c_1,\ldots,c_w\) to be on some {\it left} side and the chain for \(u\) to be on some {\it right} side. 
    Let \(A= \set{a^L_a,a^L_b,a^R_a,a^R_b, a_g}\) with the preconditions and effects of these actions described in \Cref{tab:actions_wordshuffle_verification}.
    \begin{table}
        \centering
        \begin{tabular}{|c|c|c|c|} 
            \hline
            \diagbox{Action}{Prop. set} & $\pre$ & $\del$ & $\add$ \\ 
            \hline
            $a^L_a$ & $\set{\text{LEFT}}$ & $\set{\text{LEFT}}$ & $\set{a}$ \\ 
            \hline
            $a^L_b$ & $\set{\text{LEFT}}$ & $\set{\text{LEFT}}$ & $\set{b}$ \\ 
            \hline
            $a^R_a$ & $\set{a}$ & $\set{a}$ & $\set{\text{LEFT}}$ \\ 
            \hline
            $a^R_b$ & $\set{b}$ & $\set{b}$ & $\set{\text{LEFT}}$ \\ 
            \hline
            $a_g$ & $\emptyset$ & $\emptyset$ & $\set{\text{GOAL}}$ \\ 
            \hline
        \end{tabular}
        \caption{The actions in the proof of Thm.~\ref{thm:target_cover_chains_hard}.}\label{tab:actions_wordshuffle_verification}
    \end{table}
    If the \(j^{\text{th}}\) letter of \(u\) is an \(a\), then we set \(\alpha(t_{u,j})=a^R_a\), and otherwise, \(\alpha(t_{u,j})=a^R_b\).
    Similarly, for each \(i\in[w]\), if the \(j^{\text{th}}\) letter in \(c_i\) is an \(a\), then we set \(\alpha(t_{c_i,j})=a^L_a\), and otherwise, \(\alpha(t_{u,j})=a^L_b\).
    Further, let \(\alpha(t_g) = a_g\) and \(s_0=\set{\text{LEFT}}\).
    For \ptarget, we let \(s_g = \set{\text{GOAL}}\).

    We prove that these instances are yes-instances if and only if the given \textsc{Shuffle Product} instance is a yes-instance. 
    Suppose \(u\) is a shuffle product of \(c_1,\ldots,c_w\).
    Then, each letter in the words \(c_1,\ldots,c_w\) is assigned a unique position with a matching letter in \(u\) such that the assigned positions of the letters in each word \(c_i\), \(i\in[w]\), are in order.
    Let \(t_{\ell}\) refer to the task that corresponds to the letter that is assigned to position \(\ell\). 
    Consider the task sequence \(t_1, t_{u,1}, t_2, t_{u,2}, \ldots, t_{|u|}, t_{u,|u|}, t_g\).
    This sequence does not conflict with \(\prec^+\).
    Further, the action of each task is executable in its respective state since the state alternates between \(\set{\text{LEFT}}\) and either \(\set{a}\) or \(\set{b}\), depending on which letter was seen last.
    As the letters in \(c_1,\ldots,c_w\) are assigned to positions with a matching letter in \(u\), the task \(t_{u,j}\) is always executable after the execution of \(t_{\ell}\), for all \(\ell\in [|u|]\), and in turn enables the execution of \(t_{\ell+1}\).
    After the execution of \(t_g\), the target state \(s_g\) is reached and all the tasks are executed, and thus, the reduction instance is a yes-instance.

    For the other direction, suppose there is a task sequence solving the \ptarget or \pexist instance. 
    Consider the first part of the solution sequence up to the task \(t_g\).
    All the tasks in \(T_u\) are executed before \(t_g\), and \(\prec^+\) ensures that they are executed in the order of their corresponding letters in \(u\).
    Note that, by the effects and preconditions of the actions, the sequence starts with a task in \(T_c\) and alternates between tasks in \(T_c\) and \(T_u\). 
    Further, a task in \(T_c\) is always followed by a task in \(T_u\) that corresponds to the same letter.
    Thus, listing all the letters in the words \(c_1,\ldots,c_w\) in the order their corresponding tasks are executed in the task sequence, gives \(u\).
    As \(\prec^+\) ensures that the letters from each word \(c_i\), \(i\in[w]\), are listed in the order they appear in \(c_i\), it is a yes-instance of \textsc{Shuffle Product}.

    Note that the only states in the state transition graph are \(\set{\text{LEFT}}\), \(\set{a}\), \(\set{b}\), and \(\set{\text{LEFT}, \text{GOAL}}\).
    Also, the instance is constructed in polynomial time and has \(w+1\) chains.
\end{proof}

\paragraph*{Action Executability on primitive Antichains (Thm. \ref{thm:includeaction_fpt_empty}).}

Let $k = O(1)$ be the number of states in the state transition graph.
We say that two actions $a_1,a_2$ are \emph{equivalent} if they are executable from the same states and lead to the same state, i.e., for all states $s,s'$, there is an edge (with label) $a_1$ from $s$ to $s'$ if and only if there is an edge $a_2$ from $s$ to $s'$ in the state transition graph.
We note that there are at most \((k+1)^k\) equivalence classes because each class is identified by a vector \(D\) of length \(k=|\calS|\) over \(\calS \cup \set{\emptyset}\). 
Thereby, \(D[i]=s\) is to be interpreted as the equivalence class transforms the \(i^{\text{th}}\) state into state \(s\), and \(D[i]=\emptyset\) implies the class does not transform the \(i^{\text{th}}\) state into anything as its preconditions are not fulfilled. 
This allows us to define a reduction rule reducing the number of different actions. 
Reduction rules are \emph{safe} if they preserve yes- and no-instances.

For a primitive task network \(tn=(T,\prec^+,\alpha)\) and action \(a\in A\), let \(T_a = \set{t\in T\mid \alpha(t)=a}\) and \(\max(a) = |T_a|\).
For \pcover instances with action multiset \(S\), let \(\min(a)\) be the number of occurrences of \(a\) in \(S\).
A primitive \pcover instance is a yes-instance if and only if there is a task sequence that respects $\prec^+$, is executable in \(s_0\) and in which each action \(a\) occurs at least \(\min(a)\) and at most \(\max(a)\) times. 
We now provide two reduction rules for \pcover.

\paragraph*{Reduction Rule R0:}
For a primitive \pcover instance, if there is an action \(a\in A\) such that \(\min(a) > \max(a)\), then output ``no''.

\paragraph*{Reduction Rule R1:} 
Consider a primitive \pcover instance where Reduction Rule R0 cannot be applied. Let \(a_1,a_2\in A\) with \(a_1\neq a_2\) be equivalent actions. 
If, for all \(t_1\in T_{a_1}, t_2\in T_{a_2}, t\in T\), the following holds:

\begin{itemize}
\item \((t,t_1)\in \prec^+\) if and only if \((t,t_2)\in \prec^+\) and 
\item \((t_1,t)\in \prec^+\) if and only if \((t_2,t)\in \prec^+\),
\end{itemize}
then change \(\alpha\) such that, for all \(t_2\in T_{a_2}\), now \(\alpha(t_2)=a_1\), and
replace all occurrences of \(a_2\) in \(S\) by \(a_1\).

\begin{claim}
    \textup{Reduction Rule R1} is safe.
\end{claim}
\begin{proof}
    Let \(\max_r(a)\) and \(\min_r(a)\) be defined as above, but for the reduced instance.
    Hence, \(\max_r(a_1) = \max(a_1) + \max(a_2)\) and \(\min_r(a_1) = \min(a_1) + \min(a_2)\).
    Thus, any task sequence that solves the original instance also solves the reduced instance.  
    Now assume there is a task sequence that solves the reduced instance.
    If the same task sequence solves the original instance, we are done. 
    Otherwise, let \(x_1\) and \(x_2\) be the number of occurrences of \(a_1\)- and \(a_2\)-tasks, respectively, in the sequence resolving the reduced instance according to the original mapping \(\alpha\).
    Note that \(\max_r(a_1)=\max(a_1)+\max(a_2)\ge x_1+x_2 \ge \min(a_1)+\min(a_2)\).
    Due to the equivalence in their actions and neighborhoods in \(\prec^+\), \(a_1\)- and \(a_2\)-tasks can be used interchangeably in the sequence. Thus, with the above inequalities, we can find a valid solution for the original instance, i.e., one in which $\max(a_1)\ge x_1 \ge \min(a_1)$ and $\max(a_2)\ge x_2 \ge \min(a_2)$.
\end{proof}

Note that an instance of \pcover with $|S| > |T|$ is a trivial no-instance, so we assume $|S| \le |T|$.

\begin{proof}[\textbf{Proof of \cref{thm:includeaction_fpt_empty}}]
    We begin by exhaustively applying \textup{Reduction Rules R0 and R1}, which takes at most $\bigoh(|\I|^3)$ time. Afterwards, we have either correctly solved the instance, or---since the network is an antichain---obtained a new equivalent instance where the number of actions is bounded by a function of the number $k$ of states, in particular $|A|\leq (k+1)^k$.
    
    The algorithm now branches over the at most \(2^k \cdot |A|^k\) simple edge paths in \(\G\) that start at \(s_0\).
        Let \(C\) be the set of distinct simple edge cycles in \(\G\), where two cycles are distinct if and only if they do not contain the same vertices.
        Then, \(|C| \le 2^k \cdot |A|^k\). 
        The algorithm branches on all possible subsets \(C' \subseteq C\) to decide which of these cycles are used at least once by a solution sequence.
        To ensure there is a walk, set all cycles in \(C'\) as unmarked and exhaustively mark cycles that share a vertex with the chosen path or a marked cycle.
        Proceed with \(C'\) only if all cycles are marked by this procedure.
        
        We write an \emph{Integer Linear Program} (ILP) as follows. Let the ILP have a variable \(x_c\) for every \(c\in C'\). 
        For each action \(a\) and \(c\in C'\), let \(c(a)\) denote the number of occurrences of \(a\) in \(c\).
        Further, let \(p(a)\) denote the number of occurrences of \(a\) in the chosen path.
        For each \(a\in A\), add constraints to require
        \begin{align*}
            \min(a) \le p(a) + c_1(a)x_{c_1} + \dots c_{|C'|}(a)x_{c_{|C'|}} &\le \max(a),
        \end{align*}
        where \(c_1,\ldots,c_{|C'|}\) are the elements of \(C'\).
        For all \(c\in C'\), add a constraint \(x_{c} \ge 1\) to ensure each of the selected cycles is walked at least once.
        The given instance is a yes-instance if and only if there is a solution to the ILP.
    
        If the ILP has a solution, construct a plan \(p\) as follows. 
        Start with the chosen path.
        For each cycle \(c\in C'\), add \(x_c\) repetitions of that cycle at the first time a vertex of \(c\) appears in the path.
        As the cycles end in the same state they started in, the resulting action sequence is executable.
        Possibly, some cycles cannot be added directly, but only after other cycles have been added.
        Nevertheless, as all cycles in \(C'\) were marked by the algorithm, and the ILP ensures that each cycle in \(C'\) is executed at least once, all cycles can be added in this way.
        Further, the ILP ensures that each action \(a\) appears at least \(\min(a)\) and at most \(\max(a)\) times in the resulting executable action sequence.
        Thus, the given instance is a yes-instance.
    
        For the other direction, suppose there is an executable plan \(p\) such that each action \(a\) occurs at least \(\min(a)\) and at most \(\max(a)\) times in \(p\).
        Consider an initially empty set of cycles \(C^*\).
        Then, find a substring of \(p\) that describes a simple edge cycle in \(\G\). 
        Add this cycle to \(C^*\) while disregarding its starting vertex and its direction, remove the respective substring from \(p\), and repeat until no more such cycles exist.
        Observe that \(p\) then describes a simple path in \(\G\), and \(C^*\) exclusively contains simple edge cycles in \(\G\).
        There is a branch of the algorithm that chooses this simple path and has \(C' = C^*\).
        Further, as the cycles were iteratively removed from a walk in the state transition graph, the algorithm will properly mark all cycles in \(C'\).
        Also, the ILP has a valid solution by, for each cycle \(c\in C'\), setting \(x_c\) to the number of times this cycle was removed from \(p\).  
        Thus, in this case the algorithm will decide that the given instance is a yes-instance.
      
        Computing \(\min(a)\) and \(\max(a)\) for all \(a\in A\) takes \bigO{|T|+|S|} time. 
        Combining the choices for the path and \(C'\) gives at most 
        \(2^k\cdot |A|^k \cdot 2^{2^k\cdot |A|^k}\) branches.
        Testing a set \(C'\) by the marking procedure takes \bigO{|C|^2k} time, which is dominated by the time to solve the ILP.
        The ILP has at most \(|C|\leq 2^{k} \cdot |A|^k\) variables and \(2|A| + |C|=\bigO{2^k\cdot |A|^k}\) constraints, in which the absolute values of all coefficients are bounded by \(|T|\).
        By the result of \cite{Lenstra_1983} and its subsequent improvements \cite{Frank_Tardos_1987,Kannan_1987} an ILP instance $\I_{\text{ILP}}$ with \(n\) variables can be solved in \bigO{n^{O(n)}\cdot |\I_{\text{ILP}}|} time.
        We note that, for the given ILP, $|I_{\text{ILP}}| = \bigO{2^k\cdot |A|^k \cdot \log(|T|)}$.
        Thus, solving the ILP takes 
        \bigO{(2^{k} \cdot |A|^k)^{O(2^{k} \cdot |A|^k)} \cdot 2^{2k} \cdot
             |A|^{2k} \cdot \log(|T|)} time.
        Hence, the total runtime can be upper-bounded by
        \begin{align*}
            \mathcal{O}\Big(2^k\cdot |A|^k \cdot 2^{2^k\cdot |A|^k}\cdot (2^{k} \cdot |A|^k)^{O(2^{k} \cdot |A|^k)} \cdot 2^{2k} \cdot
            |A|^{2k} \cdot \\ \log(|T|)+|T|+|S|\Big)\\
            =\bigO{2^{2^k\cdot |A|^k}\cdot (2|A|)^{O(k \cdot 2^{k} \cdot |A|^k)} \cdot \log(|T|)+|T|}.
        \end{align*}
    Recalling the bound on $|A|$ in the first paragraph of the proof, we obtain that any primitive instance of \pcover with \(\prec^+ = \emptyset\) can be solved in time $\bigoh(f(k)\cdot \log(|\I|)+|\I|^3)$, for some computable function $f$.
        \end{proof}

        \begin{proof}[\textbf{Proof of \cref{thm:planverif_powidth_xp}}]
            Let $w$ be the generalized partial order width of the input network, $U$ the set of isolated elements in $T$, and recall that the set $V$ of all remaining elements in $T$ induces a partial order of width at most $w$.
                First, check whether \(p\) is executable from $s_0$.
                If not, then return that it is a no-instance.
                Otherwise, decompose \(D_\prec[V]\) into \(w\) chains \(c_1, \ldots, c_w\) in \bigO{|V|^{2.5}} time~\cite{Fellows_McCartin_2003}.
                For \(0\le i \le |T|\), let \(R_i[h_1,\ldots,h_w]\) be a boolean variable that is \True if and only if the sequence of the first \(i\) actions in \(p\) is executable in \(tn\) such that exactly the first \(h_j\) tasks are used in each chain \(c_j, j\in [w]\).
                We compute this variable for all possible assignments via dynamic programming.
                Initialize all values to \False except for \(R_0[0,\ldots,0]\), which is \True.
                This completes the calculations for \(i=0\). 
                Assume that the calculations for some arbitrary, but fixed \(i\) are complete. 
            
                Then, update the values for \(i+1\) as follows.
                Iterate over all assignments \(R_i[h_1,\ldots,h_w]\) that are \True.
                Let \(T'\subseteq U\cup V\) be such that \(t\in T'\) if and only if \(\alpha(t) = p[i+1]\) and, for each chain \(c_j\), all the predecessors of \(t\) in \(\prec^+\) in \(c_j\) are among the first \(h_j\) tasks in \(c_j\).
                Intuitively, \(T'\) contains the tasks that match the \((i+1)^\text{th}\) action in the plan and have that each of their predecessors in $\prec^+$ is among the first \(h_j\) tasks of a chain \(c_j\).
                Further, let \(V'\subseteq V\) be such that \(t\in V'\) if and only if \(\alpha(t) = p[i+1]\) and \(t\) is among the first \(h_j\) tasks in its chain \(c_j\), that is, \(V'\) contains the already executed tasks of the current action from the chains.
                Let \(o\) denote the number of occurrences of the action \(p[i+1]\) in the first \(i\) elements of the plan \(p\).
                If \(|U\cap T'| > o - |V'|\), then set \(R_{i+1}[h_1,\ldots,h_w]\) to be \True.
                Intuitively, this means that the next task can be taken from \(U\) since there is an available task in \(U\) that is not already required earlier in the plan.
                Also, the \((i+1)^\text{th}\) task may come from a chain.
                Thus, for each chain \(c_j, j\in[w],\) such that the \((h_j+1)^\text{th}\) task is in \(T'\), set \(R_{i+1}[h_1,\ldots,h_{j-1},h_j+1,h_{j+1},\ldots,h_w]\) to \True.
                Continue until the computations for \(i=|T|\) are completed.
                Then, the given instance is a yes-instance if and only if there is \(R_{|T|}[h_1,\ldots,h_w]\) that is \True.
            
                We show by induction over \(i\) that the algorithm sets all variables correctly.
                For \(i=0\), the algorithm sets all \(R_0[h_1,\ldots,h_w]\) correctly as the only way to build an empty plan is to use no tasks at all.
                Now, assume that the algorithm correctly assigned all variables up to an arbitrary, but fixed~\(i\).
                
                If the algorithm sets some value \(R_{i+1}[h_1,\ldots,h_w]\) to \True, then we distinguish between two cases.
                First, assume it does so starting from \(R_{i}[h_1,\ldots,h_w]\) being \True.
                Due to the induction hypothesis, there is an task sequence respecting $\prec^+$ that uses the first \(h_j\) tasks of each chain \(c_j\), \(j\in[w]\), and whose actions correspond to the first \(i\) actions of \(p\).
                Action \(p[i+1]\) occurs \(o\) times among these. 
                As \(|V'|\) contains all used tasks with this action from the chains, exactly \(o-|V'|\) executed tasks of this action are in \(U\).
                As \(T'\cap U\) contains all executable tasks in \(U\), if \(|T'\cap U| > o-|V|\), then there is a task \(t\in U\) that is executable for action \(p[i+1]\) and has not already been used. 
                Thus, it is correct to set \(R_{i+1}[h_1,\ldots,h_w]\) to \True.
                In the other case, there is \(j\in[w]\) such that \(R_{i}[h_1,\ldots,h_{j-1},h_j-1,h_{j+1},\ldots,h_w]\) is \True and the \(h_j^\text{th}\) action of the chain \(c_j\) is executable.
                Thus, in this case it is also correct to set \(R_{i+1}[h_1,\ldots,h_w]\) to \True.
            
                For the other direction, suppose there is a sequence \(\bar{t}\) of \(i\) distinct tasks in \(T\) of the required actions that uses the first \(h_j\) tasks of each chain \(c_j\), that is, \(R_{i}[h_1,\ldots,h_w]\) should be \True.
                Let the last task in this sequence be \(t\) and let \(\bar{t}'\) denote the sequence \(\bar{t}\) without \(t\).
                If \(t\in U\), then the existence of \(\bar{t}'\) implies that \(R_{i-1}[h_1,\ldots,h_w]\) is \True, and in the computations starting from there we have \(|T'\cap U| > o-|V|\).
                If \(t\in V\), then \(t\) is the \(h_j^\text{th}\) task of a chain \(c_j\), \(R_{i-1}[h_1,\ldots,h_{j-1},h_j-1,h_{j+1},\ldots,h_w]\) is \True, and in the computations starting from there we have \(t\in T'\).
                In both cases, the algorithm correctly sets \(R_{i}[h_1,\ldots,h_w]\) to \True.
            
                This procedure sets \bigO{|T|\cdot |T|^w} values to \True.
                For each of these, computing \(T',V',\) and \(o\) takes \bigO{|T|+|{\prec}|} time and it can update at most \(w+1\) variables.
                Hence, setting all variables takes \(\bigO{|T|^{w+1}(|T|+|{\prec}|+w+1)} = \bigO{|T|^{w+1}(|T|+|{\prec}|)}\) time, which dominates the \bigO{|V|^{2.5}} time it takes to decompose \(D_\prec[V]\) into $w$ chains. In particular, the running time can be upper-bounded by $|\I|^{\bigoh(w)}$.
            \end{proof}
    
\begin{proof}[\textbf{Proof of \cref{thm:target_cover_powidth_xp}}]
        For an \pcover instance, first test whether Reduction Rule R0 can be applied, and if yes we use it to terminate.
        Let $k$ be the number of states and $w$ the generalized partial order width of the input instance. Set $U$ to be the set of isolated elements in $T$, and recall that the set $V$ of all remaining elements in $T$ can be decomposed into \(w\) chains \(c_1, \ldots, c_w\).
        This takes \bigO{|V|^{2.5}} time~\cite{Fellows_McCartin_2003}.
        Let \(A=\set{a_1,\ldots,a_{|A|}}\) and let \(E=\set{e_1,\ldots,e_{|E|}}\) denote the set of equivalence classes among all actions, where action equivalence is defined as in the proof of \cref{thm:includeaction_fpt_empty}.
        For each state \(s\) in the state transition graph, let \(R_s[h_1,\ldots,h_w][r_1,\ldots,r_{|E|}]\) be a boolean variable that is \True if and only if there is a sequence of tasks that respects $\prec^+$, is executable from \(s_0\), yields state~\(s\) and employs exactly the first \(h_j\) tasks of each chain \(c_j\), \(j\in [w]\), and contains exactly \(r_i\) tasks from \(U\) with action equivalence class \(e_i\), for all \(i\in [|E|]\).

        Assuming the variable is set correctly for all possible assignments, the instance can be decided as follows.
        For \ptarget, we have a yes-instance if and only if there are $h_1,\ldots,h_w, r_1,\ldots, r_w$ such that $R_{s_g}[h_1,\ldots,h_w][r_1,\ldots,r_{|E|}]$ is \True.
        For \pcover, for a variable $R_s[h_1,\ldots,h_w][r_1,\ldots,r_{|E|}]$ let \(c(a)\) denote the number of tasks in \(V\) with action \(a\) that are among the first \(h_j\) tasks of a chain \(c_j\).
        We have a yes-instance if and only if there is a \True variable such that, for each action $a$
        there are at least \(\min(a) - c(a)\) tasks of action \(a\) in \(U\) and
        for each action equivalence class \(e_i\in E\),
         we have 
         \(\sum_{a\in e_i} \max\set{0,\min(a) - c(a)} \le r_i\).
          
        We compute all variable values in a dynamic programming manner.
        Initialize all values to \False except for \(R_{s_0}[0,\ldots,0][0,\ldots,0]\), which is initialized to \True.
        Whenever a variable \(R_s[h_1,\ldots,h_w][r_1,\ldots,r_{|E|}]\) that was \False before is set to \True (including the initial one), also perform the following update step.
        Let \(T'\subseteq T\) be such that \(t\in T'\) if, for each chain \(c_j\), all the predecessors of \(t\) in \(\prec^+\) in \(c_j\) are among the first \(h_j\) tasks in \(c_j\).
        Consider all the states \(s'\) that are adjacent to \(s\) in the state transition graph \(\G\), including $s$ itself if it has a self-loop.
        If there is \(j\in[w]\) such that the \((h_j+1)^\text{th}\) task of the chain \(c_j\) is in \(T'\), has action \(a_i\), and \(a_i\) connects \(s\) and \(s'\) in \(\G\), then set 
        \(R_{s'}[h_1,\ldots,h_{j-1},h_j+1,h_{j+1},\ldots,h_w][r_1,\ldots,r_{|E|}]\) 
        to \True.
        This corresponds to taking the next task from a chain and adding it to a given sequence. 
        Another update rule corresponds to adding a task from \(U\) instead.
        If there is an action equivalence class \(e_i\) that connects \(s\) and \(s'\) in \(\G\), then define the following. 
        Let \(T_i = \set{t\in T\mid \alpha(t) \in e_i}\) and let \(v_i\) denote the number of tasks in \(V\cap T_i\) that are among the first \(h_j\) tasks of a chain \(c_j\).
        If \(r_i - v_i < |T_i\cap U|\), then set 
        \(R_{s'}[h_1,\ldots,h_w][r_1,\ldots,r_{i-1},r_i+1,r_{i+1},\ldots,r_{|E|}]\) 
        to \True.
        Once no more variables are updated and set to \True, the computation is complete.
    
        We show that the algorithm sets all variables correctly by an inductive argument over \(z = \sum_{i=1}^{w} h_i + \sum_{i=1}^{|E|} r_i\).
        When one \True variable sets another variable to \True, this increases this sum by exactly 1.
        For \(z=0\), the algorithm sets all variables correctly as the only way to build an empty plan is to take no tasks from any chain and remain in the initial state, and thus, the only \True variable in this case is \(R_{s_0}[0,\ldots,0][0,\ldots,0]\).
        Now, assume that the algorithm correctly assigned all variables up to some arbitrary, but fixed \(z\).
        Consider a variable \(R_{s}[h_1,\ldots,h_w][r_1,\ldots,r_w]\) with $\sum_{i=1}^{w} h_i + \sum_{i=1}^{|E|} r_i = z+1$.
        
        If the algorithm sets the variable to be \True, we distinguish between two cases.
        First, assume it does so starting from some \True variable \(R_{s'}[h_1,\ldots,h_w][r_1,\ldots,r_{i-1},r_i-1,r_{i+1},\ldots,r_{|E|}]\) with \(i\in[|E|]\).
        Then, \(a_i\) is executable in state \(s'\) and transforms it into \(s\).
        Further, due to the induction hypothesis, there is an executable task sequence that uses the first \(h_j\) tasks of each chain \(c_j\), \(j\in[w]\), and consists of \(r_\ell\) tasks of action \(a_\ell\), for each \(\ell\in [|A|]\setminus\set{i}\), and \(r_i-1\) tasks of action \(a_i\).
        Thus, it contains \(r_i - v_i\) tasks in \(U\) with action equivalence class \(e_i\).
        As there are more than \(r_i - v_i\) tasks of that equivalence class in \(U\) and their action is executable in the current state, one of the remaining tasks of action \(a_i\) can be appended to the sequence.
        The resulting sequence is executable from \(s_0\), respects \(\prec^+\), and uses the same number of tasks of each chain and action, except one more task with action \(a_i\). 
        Thus, \(R_{s'}[h_1,\ldots,h_w][r_1,\ldots,r_w]\) is indeed \True.

        In the other case, it sets the variable to \True starting from some \(R_{s'}[h_1,\ldots,h_{j-1},h_j-1,h_{j+1},h_w][r_1,\ldots,r_{i-1},r_i-1,r_{i+1},\ldots,r_{|E|}]\) being \True with \(j\in [w]\), \(i\in[|E|]\).
        Then, due to the induction hypothesis, there is a task sequence corresponding to an executable action sequence that uses the first \(h_{\ell'}\) tasks of each chain \(c_{\ell'}\), for each \(\ell'\in[w]\setminus \{j\}\), but only uses the first \(h_j-1\) tasks of the chain \(c_j\), and it consists of \(r_\ell\) tasks of action equivalence class \(r_\ell\), for each \(\ell\in [|E|]\setminus\set{i}\), and \(r_i-1\) tasks of action equivalence class \(e_i\). 
        Further, the \(h_j^\text{th}\) task of chain \(h_j\) has action equivalence class \(e_i\), transforms state \(s'\) into \(s\), and appending it to the sequence respects $\prec^+$.
        Thus, in this case it is also correct to set the variable to \True.
    
        For the other direction, suppose the variable is supposed to be \True and let $\bar{t}$ be a task sequence witnessing its correctness.
        Consider the variable corresponding to the state, and values of the $h_i$ and $r_j$ after executing all but the last element of $\bar{t}$. 
        As this shortened task sequence is a prefix of $\bar{t}$, its corresponding variable is supposed to be \True and the algorithm correctly identifies this by the induction hypothesis.
        When the algorithm sets this variable to be \True, in the following update steps it considers all possible next tasks in the sequence, including the last element of $\bar{t}$, and hence correctly sets the value of that variable to be \True.
        
        Computing the equivalence classes takes \bigO{|\I|^3} time.
        As there are \(|E|\le (k+1)^k\) equivalence classes, at most 
        \bigO{|T|^w \cdot k \cdot |T|^{(k+1)^k}} 
        variables are considered. 
        For each of these, computing all \(v_i\) takes \bigO{|T|} time by iterating through all chains and it can update at most \(w+|E|\) variables. 
        Further, it takes \bigO{|T|} time to check whether a variable yields a yes-instance.
        Hence, the algorithm takes 
        \(\bigO{|T|^{w+(k+1)^k} \cdot k \cdot (|T||+w+|E|) + |V|^{2.5}+|{\prec}|} 
        = \bigO{|T|^{w+(k+1)^k+1} \cdot k}\) time.
    \end{proof} 

    \begin{proof}[\textbf{Proof of \cref{thm:meta_theorem_lowerbound}}]
        We reduce from the \NP-hard \textsc{Multicolored-Clique} problem, which asks whether a given $k$-partite graph \(G=(V,E)\) contains a $k$-clique~\cite{FELLOWS2009}.
            Let \(tn(p,d,a)\) with \(p,d,a\subseteq F\) refer to a task network with a single task \(t\) with action \(\alpha(t)\) such that \(\pre(\alpha(t))=p\), \(\del(\alpha(t))=d\), and \(\add(\alpha(t))=a\).
            For each vertex \(v_i\), \(i\in |V|\), create a compound task \(V_i\), and, for each edge \(e_i\), \(i\in |E|\), create a compound task \(E_i\).  
            Create a primitive task \(t_g\) with action \(a_g\) such that \(\pre(a_g)=\set{\text{EDGE}_{j,j'}\mid j,j'\in [k], j < j'}\), \(\del(a_g)=\emptyset\), and \(\add(a_g)=\set{\text{GOAL}}\).
            The initial network \(tn\) contains exactly these tasks in a total order where \((V_i,V_{i+1})\in \prec\) for all \(i\in[|V|-1]\), \((V_{|V|}, E_1)\in \prec\), \((E_i,E_{i+1})\in \prec\) for all \(i\in[|E|-1]\), and \((E_{|E|}, t_g)\in \prec\).
            Each vertex and edge compound task can be decomposed into a primitive task with no effect, that is, 
            \((V_i,tn(\emptyset,\emptyset,\emptyset))\in M\) for all \(i\in [|V|]\), and 
            \((E_i,tn(\emptyset,\emptyset,\emptyset))\in M\) for all \(i\in [|E|]\).
            This corresponds to not selecting respective vertices and edges as part of the clique.
            Further, a vertex can be chosen for color \(j\) if no other vertex has been chosen for this color (represented by having \(\text{COLOR}_j\) in the state).
            Thus, \((V_i, tn(\set{\text{COLOR}_j}, \set{\text{COLOR}_j}, \set{v_i}))\in M\) for \(i\in [|V|]\) and the color \(j\) of the vertex \(v_i\).
            An edge can be chosen as the \(j^{\text{th}}\) edge in the clique if both its vertices are chosen.
            Thus, for all \(i\in [|E|]\), if \(e_i = (v,v')\), \(v\) has color \(j\), and \(v'\) has color \(j'\), then \((E_i, tn(\set{v, v'}, \emptyset, \set{\text{EDGE}_{j,j'}}))\in M\).
            The construction is concluded by setting the initial state \(s_0 = \set{\text{COLOR}_i\mid i\in [k]}\), setting \(S=\set{a_g}\) for \pcover, and setting \(s_g = \set{\text{GOAL}}\) for \ptarget.
            The construction takes polynomial time.
        
            Suppose there is a multicolored clique of size \(k\). 
            Decompose according to the above interpretation.
            We get a primitive network with a total order on the tasks.
            The corresponding action sequence is executable from \(s_0\) and executes all tasks including \(t_g\), witnessing a yes-instance for \pexist, \pcover, and \ptarget.
        
            For the other direction, suppose any one of the three HTN instances has a solution. 
            This implies that the task \(t_g\) can be executed.
            The precondition of \(t_g\) ensures that an edge is selected for each pair of colors.
            The preconditions on the edge tasks ensure that all the vertices incident to the selected edges have been selected.
            The preconditions on the vertex tasks ensure that at most 1 vertex of each color is selected.
            Thus, we have a yes-instance of \textsc{\(k\)-Multicolored-Clique}.
          
            Note that for the constructed instances, \(\compd = \comps = 1\) and \(\compc = 2\), proving the case of non-constant $\compn$.
            For non-constant $\comps$, we let the initial network contain a single compound task that is decomposed into a network with the above structure. Then, \(\compc=2\), \(\compd=2\), and \(\compn=1\).
            For non-constant $\compd$, we let the initial network only contain the compound task \(V_1\).
            Then, we adapt all the decomposition methods. 
            For each method \((t, tn_m)\in M\), the network \(tn_m\) consists of a single primitive task \(t_m\). 
            Let \(t'\) denote the task identifier that followed \(t\) in the total order in the initial network. 
            Instead of decomposing into a single primitive task \(t_m\), let the method decompose into a network with tasks \(t_m\) and \(t'\) with \((t_m, t')\in \prec^+\).
            Note that the exact same set of primitive networks can be created as in the original construction, but now \(\compd\) is non-constant, \(\compc=\comps = 2\), and \(\compn=1\).
            
            All initial task networks have (generalized) partial order width 1 or 0 as they consist of a chain or a single vertex.
        \end{proof}

 \begin{proof}[\textbf{Proof of \cref{thm:planverif_vertexcover}}]
    We have a yes-instance if the input plan \(p\) is executable in $s_0$ and each \(t\in T\) can be assigned a unique index in the sequence such that the assignment according to \(\alpha\) matches the plan and the imposed order does not conflict with \(\prec^+\).
    Observe that simulating the action sequence suffices to check whether all preconditions are met at the respective step, so it remains to handle the assignment of tasks to indices.
    We branch over all possible orderings of the vertices in \(V\).
    Orderings that do not respect an edge in \(\prec^+ \cap (V\times V)\) are discarded immediately.
    For \(i\in [|V|]\), let \(v_i\) denote the \(i^{\text{th}}\) task in the ordering. 
    Then, for all \(t\in T\setminus V\), let the priority of \(t\) be defined as \(\psi(t)=\min_{ (t,v_i)\in \prec} i\) and \(\psi(t)=\infty\) if there is no such edge. 
    For all \(v_i\in V\), let \(\psi(v_i)=i + 0.5\). 
    We now iterate through \(p\) and assign tasks to each index in a greedy manner.
    Assume we have already assigned tasks to the first \(i-1\) actions for some \(i\in [|T|]\).
    Then, we say that \(t\in T\) is executable as the \(i^{\text{th}}\) action if \(\alpha(t)=p[i]\), it is not yet assigned an index, and all its predecessors in \(\prec^+\) are assigned to indices less than~\(i\).
    Among all executable tasks at the index \(i\), we choose the one with the lowest priority according to \(\psi\).
    If there is no executable task for some index, then we return that this branch does not yield a yes-instance. 
    The algorithm returns that it is a yes-instance if any branch manages to assign all tasks to~\(p\), and rejects otherwise.\\
    \indent    
    If a branch assigns each task to an index in the plan, it is indeed a yes-instance as the algorithm only assigns executable tasks at each step.
    For the other direction, assume there is a bijective assignment \(a^*\colon T\rightarrow[|T|]\) such that, for all \(t\in T\), we have \(\alpha(t)=p[a^*(t)]\) and the corresponding ordering of the tasks according to \(a^*\) does not conflict with~\(\prec^+\). 
    Consider the branch in which the tasks in \(V\) are ordered according to~\(a^*\).
    As \(a^*\) is a valid assignment, this branch is not immediately discarded.
    Further, this branch successfully assigns a task to each index by the following inductive argument. 
    We claim that, for each \(0\le i\le |T|\), there is a valid assignment that assigns the indices 1 to \(i\) exactly like the algorithm does.
    As \(a^*\) is a valid assignment, this holds for \(i=0\).
    Suppose the claim holds for an arbitrary but fixed \(i\) and the algorithm is about to assign index \(i+1\). 
    Then, by the induction hypothesis, there is a valid assignment \(a_i^*\) that assigns all tasks up to \(i\) the same as the algorithm has. 
    Let \(t\) and \(t^*\) be the tasks that are assigned to \(i+1\) by the algorithm and \(a_i^*\), respectively.
    If \(t=t^*\), then the induction claim trivially holds for \(i+1\).
    Thus, assume that \(t\neq t^*\). 
    As \(a_i^*(t^*)=i+1\), \(t^*\) is executable when the algorithm assigns \(t\) to \(i+1\). 
    Since the algorithm chooses \(t\) instead, we have \(\psi(t)\le \psi(t^*)\).
    We show that \(a_{i+1}^*\), the assignment obtained from \(a_i^*\) by switching only \(t\) and \(t^*\), is a valid assignment.
    Note that the algorithm chooses \(t\) at index \(i+1\), and so, \(t\) is executable at this point. 
    It remains to show that there is no task \(t'\) such that \(a_i^*(t^*) < a_i^*(t') < a_i^*(t)\) and \((t^*,t')\in \prec^+\).
    Observe that when a task from \(V\) is executable, it is the only executable task with the lowest priority, and thus, since \(\psi(t)\le \psi(t^*)\), we get that \(t^*\notin V\).
    Hence, as \(V\) is a vertex cover, we only have to consider the case where \(t'\in V\).
	This implies that if \(\psi(t^*)=\infty\), then no conflicting edge \((t^*,t')\) exists, and so, we only have to consider the case where \(\psi(t^*)\neq \infty\), and thus, \(\psi(t)\neq \infty\).
    Since \(a_i^*\) is a valid assignment, either \(t\) is placed before \(v_{\psi(t)}\) or $t=v_{\psi(t)-0.5}$.
    Now, \(\psi(t)\le \psi(t^*)\) immediately yields that, by swapping \(t\) and \(t^*\), \(t^*\) is still placed before \(v_{\psi(t^*)}\), and hence, no conflicting edge \((t^*,t')\) exists.
    Thus, \(a_{i+1}^*\) is a valid assignment and the induction claim holds for \(i+1\).
    By the inductive argument, the algorithm finds a valid assignment if there is one, and hence, correctly solves the decision problem.\\
    \indent
    The algorithm branches over the \(|V|!\) possible orderings of tasks in the vertex cover.
    For each ordering, the priorities can be computed in \bigO{|{\prec}|+|T|} time. 
    Using a Hollow Heap \cite{Hansen_Kaplan_Tarjan_Zwick_2017} to implement a priority queue, iterating through the sequence, extracting the minimum priority executable task, and inserting all newly executable tasks to the queue takes \bigO{|{\prec}|+|T|\log |T|} time, yielding a runtime of \bigO{|V|!(|{\prec}|+|T|\log |T|)} for the entire algorithm.
\end{proof} 

\begin{proof}[\textbf{Proof of \cref{thm:reachtarget_includeaction_fpt}}]
    Let the input instance consist of the primitive network \(tn=(T, \prec^+, \alpha)\), an initial state \(s_0\), and either a target state \(s_g\) (for \ptarget) or an action set $S$ (for \pcover). For \pcover, assume that Reduction Rule R0 cannot be applied (as otherwise we can solve the instance directly). First use a fixed-parameter algorithm to compute a minimum-size vertex cover $V$ of \(G_\prec\)~\cite{DowneyFellows13}, and let $k$ denote the number of states in the state transition graph~\(\G\). 
    
    The algorithm proceeds analogously for both problems, so we provide only a single description while making explicit the few areas where the techniques diverge.
    The algorithm branches over all possibilities for the set of employed vertex cover tasks \(V'\subseteq V\), orders \(v_1,\ldots,v_{|V'|}\) of the tasks in \(V'\), and states \(s_1,\ldots,s_{|V'|}\) from which the tasks are executed, respectively. 
    This gives \bigO{2^{|V|}\cdot |V|! \cdot k^{|V|}} branches.

    For each branch, we define a {\it strong equivalence relation} on the tasks. For all \(t\in T\setminus V\), the \emph{admissible interval} is defined as \([\max\set{i \mid v_i \prec t},\min\set{i \mid t \prec v_i}-1]\).%
    \footnote{If the first set is empty, then the interval starts at 0. If the second set is empty, then it ends at \(|V'|\).} 
    Then, \(t_1,t_2\in T\) are strongly equivalent if and only if \(\alpha(t_1)\) and \(\alpha(t_2)\) are equivalent and have the same admissible interval.
    There are at most \( E = (k+1)^k \cdot (|V'|^2+1)\) strong equivalence classes.
    Define the augmented state transition graph \(\G'\) over the same set of vertices as the regular state transition graph, but let it have no edges for actions.
    Instead, let there be an edge with label \(e\) from $s$ to $s'$ for all strong equivalence classes \(e\) and states \(s,s'\) such that \(e\) transforms \(s\) into \(s'\).
    
    Then, split the branch for each of the possible \(|V'|+1\) simple edge paths in the augmented state transition graph \(\G'\) that start and end in the following vertices.
    A path \(W_0\) from \(s_0\) to \(s_1\),
    a path \(W_1\) from \(s_1\setminus \del(v_1)\cup \add(v_1)\) to \(s_2\), and so on.
    For \pcover, also branch on all possibilities for a path \(W_{|V'|}\) starting at \(s_{|V'|}\setminus \del(v_{|V'|})\cup \add(v_{|V'|})\) and ending at any state.
    For \ptarget, only consider paths \(W_{|V'|}\) that end in \(s_g\). 
    If one of these paths cannot be constructed as there is no directed path between the respective vertices, then discard the branch.
    This splits each branch into $\bigO{(2E)^{k(|V'|+1)}}$ new branches.

	Further, let \(C\) be the set of distinct simple edge cycles in \(\G\), where two cycles are distinct if and only if they do not contain the same vertices.
    Then, \(|C| \le (2E)^k\).
    For each \(0\le i\le |V'|\), we branch over all possible sets \(C_i\subseteq C\), the cycles that are walked as a detour from path \(W_i\).
    This splits each prior branch into
    \(\bigO{(2^{|C|})^{|V'|+1}} \leq \bigO{2^{(2E)^k \cdot (|V'|+1)}}\)
     new branches.
    Thus, the total number of branches is
    \[ 
        \mathcal{O}\Big(2^{|V|}\cdot |V|! \cdot k^{|V|}\cdot
         (2E)^{k(|V'|+1)}
         \cdot 
         2^{{(2E)^k \cdot (|V'|+1)}} \Big).
    \]

    Branches where there is a task \(v_i\) such that \(\pre(v_i)\nsubseteq s_i\) are discarded.
    Further, a branch is discarded if, for some \(0\le i\le |V'|\), there is an equivalence class in \(W_i\) or in a cycle \(c\in C_i\) such that \(i\) is outside the admissible interval of that class. 
    Branches where there is \(0\le i\le |V'|\) such that \(W_i\) and its associated cycles do not describe a connected component are discarded as well (as can be tested by the marking procedure described in the proof of \cref{thm:includeaction_fpt_empty}). 

    We build a solution plan that walks each of the paths \(W_i\), \(0\le i\le |V'|\), and additionally walks through the associated cycles. 
    The number of times a cycle is walked is determined by an ILP with a variable \(x_c\) for every \(c\in C\). 
    For each strong equivalence class \(e\), let \({\max}'(e)\) denote the number of tasks in \(T\setminus V\) in equivalence class \(e\).
    For all \(c\in C\), let \(c(e)\) denote the number of occurrences of \(e\) in \(c\).
    Further, let \(w_i(e)\) denote the number of occurrences of \(e\) in \(W_i\) for \(0\le i\le |V'|\).
    Note that, in any plan that is a solution to the problem instance, two strongly equivalent tasks can be swapped and the plan will still be valid.
    Thus, to ensure no task in \(T\setminus V\) is executed more than once, for each equivalence class \(e\) add a constraint
    \begin{align*}
        {\max}'(e) &\ge w_0(e)+\dots+w_{|V'|}(e)\\
         &\hphantom{\ge } 
         + c_1(e)x_{c_1} + \dots + c_{|C|}(e)x_{c_{|C|}},
    \end{align*}
    where \(c_1,\ldots,c_{|C|}\) are the elements of \(C\).
    Further, to ensure that each selected cycle is executed at least once for each path it is associated with, for all cycles \(c\in C\), add a constraint 
    \begin{equation*}
        x_{c} \ge |\set{0\le i \le |V'|\mid c \in C_i}|. 
    \end{equation*}
    For \pcover, for each action equivalence class \(e_a\), let \(\min'(e_a)\) be the number of occurrences of \(e_a\) in \(S\), and let \(w_i'(e_a)\) and \(c_i'(e_a)\) be the number occurrences of \(e_a\) in path \(W_i\) and cycle \(c_i\), respectively.
    Then, add a constraint 
    \begin{align*}
        {\min}'(e_a) &\ge w_0'(e_a)+\dots+w_{|V'|}'(e_a)\\
        &\hphantom{\ge } + c_1'(e_a)x_{c_1} + \dots + c_{|C|}'(e_a)x_{c_{|C|}}.
    \end{align*}
    For all instances, the algorithm decides that the given instance is a yes-instance if and only if there is a solution to the ILP for any non-discarded branch.

    If the ILP has a solution, then construct a sequence of equivalence classes as follows. 
    For each \(0\le i\le |V'|\), take any cycle \(c\in C_i\) that shares a vertex \(s\) with the current path \(W_i\).
    Delete \(c\) from \(C_i\) and augment \(W_i\) by inserting cycle \(c\) at that position once. 
    If there is no more \(0\le j\le |V|'\), \(j\neq i\), such that \(c\in C_j\), instead add as many repetitions of \(c\) to \(W_i\) as the number of times that cycle \(c\) was added to the path precisely \(x_c\) times in total.
    As the cycles end in the same state they started in, the resulting action sequence is executable.
    Repeat until \(C_i\) is empty.
    As \(C_i\) and \(W_i\) form a connected component, all cycles in \(C_i\) can be added this way.
    Construct a solution task sequence by concatenating the augmented paths \(W_0,\ldots, W_{|V'|}\) and, for a step by a certain equivalence class, employ any task of that class.
    For \pcover, choose tasks with a not yet covered action in \(S\) whenever possible.
    Right before each path \(W_i\),\( i\in [|V'|]\), execute task \(v_i\).
    Observe that this way the preconditions of each task are fulfilled when it is about to be executed.
    Further, branches in which there are tasks that are executed before their predecessors or after their successors in \(\prec^+\) are already discarded.
    The ILP ensures that each equivalence class~\(e\) appears at most \(\max{'}(e)\) times.
    As there are \(\max{'}(e)\) tasks of that class and they can be used interchangeably, the sequence is a valid construction.
    For \ptarget, only paths that end in \(s_g\) are considered.
    For \pcover, the additional set of constraints ensures that \(S\) is covered as there are enough tasks of each action and tasks from the same strong equivalence class can be used interchangeably.
    Thus, the given instance is a yes-instance.

    For the other direction, suppose there is a task sequence witnessing a yes-instance.
    There is a branch of the algorithm that considers exactly the set \(V'\) of vertex cover tasks used in the sequence in the order $v_1,\ldots,v_{|V|}$ and executes each task $v_i$ from state $s_i$.
    Split the sequence at the vertex cover tasks into subwalks \(W_0^*,\ldots,W_{|V|}^*\).
    Then, for all \(0\le i\le |V'|\), consider an initially empty set of cycles \(C^*_i\).
    Find any substring \(W_i^*\) that describes a simple edge cycle in \(\G'\).
    Add this cycle to \(C^*_i\) (disregarding direction and starting vertex) and remove the respective substring from \(W_i^*\).
    Repeat until no more such cycles exist.
    Observe that each \(W_i^*\) then describes a simple edge path in \(\G'\), and all \(C^*_i\) exclusively contain simple edge cycles in \(\G'\). 
    We write \(W_i'\) for these reduced \(W_i^*\).
    There is a branch of the algorithm such that \(W_i=W_i'\) and \(C_i = C_i^*\) for all \(0\le i\le |V'|\).
    This branch is not discarded because the task sequence is valid, so \(\pre(v_i)\subseteq s_i\) for all \(i\in[|V'|]\), and the paths \(W_i\) and cycles in \(C_i\) are such that they do not depend on any vertex cover task after \(v_i\), and no vertex cover task up to \(v_i\) depends on them.
    Further, as the cycles in \(C_i^*\) were iteratively removed from the path \(W_i^*\), the paths and cycles in \(W_i\) and \(C_i\) form a connected component in \(\G'\).
    For each cycle \(c\in C\), let \(x_c^*\) be the total number of times \(c\) has been removed from paths \(W_i^*\).
    Then, the ILP has a solution if \(x_c = x_c^*\) for all \(c\in C\) as follows.
    As the \(x_c^*\) are computed from a valid sequence, summing over all occurrences of an equivalence class \(e\) in paths \(W_i\) and cycles in \(C_i\) does not exceed \(\max{'}(e)\).
    Further, by construction, \(x_c^*\) is at least the number of sets \(C_i\) with \(c\in C_i\).
    In the case of an \pcover solution, the number of occurrences of each action equivalence class \(e_a\) is at least \({\min}'(e_a)\), 
    satisfying the additional constraints.
    Hence, the ILP has a solution and the algorithm correctly decides for a yes-instance.

    Applying reduction rule R0 and building the augmented state transition graph takes \bigO{|A|k+|T|} time.
    Deciding whether to discard a branch takes \bigO{|{\prec}|+|T|+|V||C|^2k} time.
    Computing \(c(e)\), \(w_i(e)\), \(\max'(e)\), \(c'(e_a)\), \(w_i'(e_a)\), \(\min'(e_a)\), and \(|\set{0\le i\le |V'|\mid c\in C_i}|\) takes $\mathcal{O}(|T|+ (|C|+|V|+1)k+ |C|\cdot(|V|+1))$ time.  
    The ILP has \(|C|\) variables and at most \(E+|C| + (k+1)^k = \bigO{E+|C|}\) constraints, in which the absolute values of all coefficients are bounded by \(\max\set{k,|T|}\).
    As noted before, an ILP instance $\I_{\text{ILP}}$ with \(n\) variables can be solved in \bigO{n^{O(n)}\cdot |\I_{\text{ILP}}|} time.
    Thus, the ILP is solved in 
    \begin{align*}
        \mathcal{O}\Big(|C|^{O(|C|)}\cdot |C|\cdot (E+|C|) \cdot \log(\max\set{(k,|T|)})\Big)
    \end{align*}
    time.
    As \(|C|\le (2E)^k\), the total runtime is at most
     \(f(k,|V|) \cdot |T|\)
    for some computable function \(f\).
    Multiplying the number of branches with the per-branch runtime yields a quantity that is at most
    $2^{k^{\mathcal{O}(k^2)} |V|^{\mathcal{O}(k)}} \cdot (|T|^2)$.
    \end{proof}

\fi

\end{document}